\documentclass[final,onefignum,onetabnum]{siamart220329}
\usepackage{cite}
\usepackage{amsmath,amssymb,amsfonts,stmaryrd}
\usepackage{algorithmic}
\usepackage{graphicx}
\usepackage{algorithm,algorithmic}

\usepackage{hyperref}
\hypersetup{colorlinks=true, linkcolor=blue, breaklinks=true, urlcolor=blue}

\usepackage{textcomp}

\usepackage{stfloats}

\def\BibTeX{{\rm B\kern-.05em{\sc i\kern-.025em b}\kern-.08em
    T\kern-.1667em\lower.7ex\hbox{E}\kern-.125emX}}
\usepackage{dsfont}
\usepackage{tikz-cd} % <-- needed for implication diagram
\usepackage{comment}
\usepackage{arydshln} %<-- for dashed line in matrix
\usepackage{multirow} %<-- table of constants
\usepackage{nicefrac, xfrac} %<-- diagonal frac for LTI th 2

\newcommand{\mcA}{\mathcal{A}}
\newcommand{\mcB}{\mathcal{B}}
\newcommand{\subscr}[2]{#1_{\textup{#2}}}

\newcommand{\osL}{\operatorname{osLip}}
\newcommand{\Lip}{\operatorname{Lip}}
\newcommand{\eps}{\epsilon}
\newcommand{\R}{\mathbb{R}}
\newcommand{\norm}[1]{\left\| #1 \right\|}

\newcommand{\lognorm}[1]{\mu(#1)}

\newcommand{\mcX}{\mathcal{X}}
\newcommand{\mcY}{\mathcal{Y}}

\newsiamremark{remark}{Remark}
\crefname{remark}{Remark}{Remarks}

\newsiamremark{problem}{Problem}
\crefname{problem}{Problem}{Problems}

\newsiamremark{assumption}{Assumption}
\crefname{assumption}{Assumption}{Assumptions}

\newcommand*{\QEDB}{\hfill\ensuremath{\square}}%  empty square
\DeclareMathAlphabet{\mymathbb}{U}{BOONDOX-ds}{m}{n}

\newcommand{\zero}{\mymathbb{0}} 

\usepackage[normalem]{ulem} % ability to strike out text from LC

\usepackage[prependcaption,colorinlistoftodos]{todonotes}
\definecolor{gnblue6}{RGB}{35,156,255} % 239cff

\allowdisplaybreaks[3]

\begin{document}

\title{Online Feedback Optimization and Singular Perturbation via
  Contraction Theory\thanks{This work was supported in part by NSF award
    1941896, NSF GRFP award DGE 2040434, and AFOSR grant FA9550-22-1-0059.
    The Authors would like to thank Dr.\ Giulia De Pasquale for the
    insightful discussions. Additionally, we acknowledge the use of a large language model
    for assistance in improving the wording and grammar of this document.}}

  % Singular Perturbation via Contraction Theory %Analysis of Singularly Perturbed Systems Through the Lens of Contraction Theory
  
\author{Liliaokeawawa Cothren\thanks{Department of Electrical, Computer,
    and Energy Engineering, University of Colorado Boulder, Boulder, CO,
    USA (\email{lico4769@colorado.edu}).} \and Francesco
  Bullo\thanks{Center for Control, Dynamical Systems, and Computation, UC
    Santa Barbara, Santa Barbara, CA 93101 USA (\email{bullo@ucsb.edu}).}
  \and Emiliano Dall'Anese\thanks{Department of Electrical and Computer Engineering, Boston University, Boston, MA 02215 USA (\email{edallane@bu.edu}). }
  }

\maketitle

% \todoinb{LC June 03, 2024: This document has been edited for typos and checked for technical mistakes. Theorems regarding the overall contractivity of the interconnected system (old Theorems 2, 4, 6), associated assumptions (Assumption 8, 14), and remarks about these results are commented out. The sub-cases for Theorem 1 are commented out. Without further editing, this leaves 10 pages total. Next, LC will reorganize the document so that we start with the OFO motivating example and then state and prove Theorem 1. BIG TO DO: reorganize the introduction to reflect these changes.}

% \todoinb{LC June 10, 2024: Updated to SIAM template (thanks FB!). Edited equation formatting for single column format (see red equations for equations that LC is unsure of how to best format still). TO DO: add (in order of preference until 22-23 pages reached) LTI section, autonomous section, Stackelberg game. Very likely to leave out the local contractivity theorem still. }

% \todoinb{15 pages: OFO motivating example, general case, proof for this, prove the OFO case... 
% New outline: OFO motivating example, general case (w sub-section/corolary for autonomous dynamics), proof for general setting (noting that just cross out zeros for autonomous case), prove the OFO case, THEN LTI setting w appendices added back in.}

\begin{abstract}
  In this paper, we provide a novel contraction-theoretic approach to
  analyze two-time scale systems, including those commonly encountered in
  Online Feedback Optimization (OFO).  Our framework endows these systems
  with several robustness properties, enabling a more comprehensive
  characterization of their behaviors.  The primary assumptions are the
  contractivity of the fast sub-system and the reduced model, along with an
  explicit upper bound on the time-scale parameter.
  For two-time scale systems subject to disturbances, we show that the
  distance between solutions of the nominal system and solutions of its
  reduced model is uniformly upper bounded by a function of contraction
  rates, Lipschitz constants, the time-scale parameter, and the variability
  of the disturbances over time.
  Applying these general results to the OFO context, we establish new
  individual tracking error bounds, showing that solutions converge to
  their time-varying optimizer, provided the plant and steady-state
  feedback controller exhibit contractivity and the controller gain is
  suitably bounded.
  Finally, we explore two special cases: for autonomous nonlinear
  systems, we derive sharper bounds than those in the general results, and
  for linear time-invariant systems, we present novel bounds based on induced matrix 
  norms and induced matrix log norms.  
\end{abstract}

\section{Introduction}
\label{sec:introduction}

\textit{Problem description and motivation}.

Singular perturbation theory is a well-established framework to analyze
multi-time scale dynamical systems.  Originating from the early works of
Tikhonov~\cite{AT:1948}, Levinson~\cite{levinson1947perturbations}, and
Vasileva~\cite{vasil1963asymptotic}, its introduction to control
engineering is primarily credited to the seminal contributions of
Kokotovi{'c} and Sannuti in the late
1960s~\cite{PK-PS:68,sannuti1969near}. Since then, singular perturbation
has become a widely used tool for the modeling, analysis, and design of
control systems~\cite{kokotovic1976singular}.
Applications of singular perturbation theory span various fields, including
power systems~\cite{chow1990singular}, power
electronics~\cite{pahlevaninezhad2013self}, robotics~\cite{JK-EAC:17},
chemical reaction networks~\cite{EDS:07}, and biomolecular
networks~\cite{DDC:13}. This theory provides powerful methods to compare
the solutions and stability properties of a two-time scale system with its
\textit{limiting system}, which consists of the \textit{reduced model} and
the \textit{boundary layer system}~\cite{PVK-HKK-JOR:99}. The reduced model
simplifies the slow dynamics by using the quasi-steady-state of the fast
sub-system, while the boundary layer system accounts for the transient
behavior of the fast variables.

% Many two-time scale example systems satisfy contractivity properties that are stronger than the usual stability assumptions in singular perturbation analysis, 
% so alternative results can be achieved that clarify their advantageous properties due to contractivity.
% \red{[can chop the first sentence.]}
% so more powerful guarantees are achievable than the current results obtained via Lyapunov theory only. 
% For example, w

% Paragraph on online feedback optimization
Online Feedback Optimization (OFO) is a powerful control framework to
design feedback controllers that guide a system (referred to as the
\textit{plant}) toward its optimal steady states. In essence, the
controller selects the optimal steady state from all possible options, as
defined by an associated optimization problem. The design of these
optimization-based control systems is heavily influenced by factors such as
the (non)convexity of cost functions, types of constraints, the
(non)linearity and stability of the system, and the suitability of
continuous-time, discrete-time, or hybrid
interconnections~\cite{MC-ED-AB:20,AH-SB-GH-FD:21,LSPL-JWSP-EM:21}.
For certain classes of unconstrained optimization problems, a common and
effective design choice is the gradient-based feedback
controller~\cite{SM-AH-SB-GH-FD:18,MC-ED-AB:20}, which uses the gradient of
convex cost functions along with a proportional constant, or \textit{gain},
that needs to be tuned. Traditional open-loop optimization methods require
the evaluation of the steady state values of the plant. However, if the
gain is sufficiently small to induce a time-scale separation between the
plant and the controller, the gradient-based OFO controller can operate
using direct measurements of the plant instead of its steady states,
thereby making the algorithm \textit{online}. Thus, the interconnected
system arising in the OFO framework serves as a prime example of a general
two-time scale system we aim to consider.

% paragraph on contraction theory
% In applications such as
Aside from optimization-based control~\cite{MC-ED-AB:20,AH-SB-GH-FD:21},
applications such as learning and control using neural
networks~\cite{HT-SJC-JJES:21,YF-TGK:96,MR-RW-IRM:20,LK-ME-JJES:22},
interaction networks~\cite{MAAR-DA-EDS:23}, hierarchical
games~\cite{LJR-SAB-SSS:16,fiez2019convergence}, and systems in Lur'e
form~\cite{MG-VA-ST-DA:23}, the two-time scale dynamical system of interest
satisfies a stronger contractivity property~\cite{WL-JJES:98,FB:24-CTDS},
instead of just the existence of a Lyapunov function.  Moreover, when
exogenous inputs and parameters are present, the system of interest often
satisfies Lipschitz conditions with respect to these inputs and parameters.
This is especially important for the OFO case, for which the vast
  majority of analyses use Lyapunov functions to guarantee
  input-to-state-stability (ISS) of the interconnected system. By using
  contractivity arguments instead of Lyapunov-based methods, we will show
  individual asymptotic bounds of the plant and the controller,
  respectively, instead of one bound on the concatenated state. 
Contraction rates and Lipschitz constants directly characterize the robust
stability of dynamical systems by, for example, serving as key quantities
for the calculation of practical stability regions, convergence times of
neighboring solutions, and tracking errors of equilibrium curves.
Additionally, contracting dynamical systems enjoy critical robustness
properties, including robustness to disturbances, incremental
input-to-state stability (ISS), and finite input-state
gain~\cite{ZA-EDS:14b,HT-SJC-JJES:21,JJ-TIF:10}.  Motivated by the
contractive properties enjoyed by many two-time scale systems, we are
interested in analyzing two-time scale systems via contraction theory to
explicitly characterize their strong stability and robustness properties.

Beyond optimization-based control~\cite{MC-ED-AB:20,AH-SB-GH-FD:21},
applications such as learning and control using neural
networks~\cite{HT-SJC-JJES:21,YF-TGK:96,MR-RW-IRM:20,LK-ME-JJES:22},
interaction networks~\cite{MAAR-DA-EDS:23}, hierarchical
games~\cite{LJR-SAB-SSS:16,fiez2019convergence}, and systems in Lur’e
form~\cite{MG-VA-ST-DA:23} often involve two-time scale dynamical systems
and often exhibit strong contractivity
properties~\cite{WL-JJES:98,FB:24-CTDS} beyond just the existence of a
Lyapunov function. Furthermore, when exogenous inputs and parameters are
present, these systems frequently satisfy Lipschitz conditions with respect
to these inputs and parameters.
This is particularly crucial for the OFO case, where most analyses utilize
Lyapunov functions to ensure input-to-state stability (ISS) of the
interconnected system. By employing contractivity arguments instead of
Lyapunov-based methods, we can demonstrate individual asymptotic bounds for
the plant and the controller, rather than a single bound on the
concatenated state.

Contraction rates and Lipschitz constants are fundamental in characterizing
the robust stability of dynamical systems, serving as key quantities for
calculating practical stability regions, convergence times of neighboring
solutions, and tracking errors of equilibrium curves. Contracting dynamical
systems also benefit from essential robustness properties, such as
robustness to disturbances, incremental input-to-state stability (iISS),
and finite input-state
gain~\cite{ZA-EDS:14b,HT-SJC-JJES:21,JJ-TIF:10,FB:24-CTDS}.  Motivated by
the contractive properties inherent in many two-time scale systems, our
interest lies in analyzing these systems through contraction theory to
explicitly characterize their strong stability and robustness properties.

\textit{Literature review}.

Within control engineering, results on the closeness of solutions of a
singularly perturbed system to its limiting system and overall stability
guarantees typically leverage Lyapunov theory. Lyapunov-based analysis
involves constructing an appropriate Lyapunov function and often requires a
priori knowledge of the attractor location. Closeness of solutions results
are categorized into two settings: the compact time interval or the
infinite time interval~\cite{PVK-HKK-JOR:99,HKK:02,FCH:66}.  For the
infinite time interval case, the reduced model and the boundary layer must
have exponentially stable equilibria, and the time-scale inducing
parameter, $\epsilon > 0$, must be sufficiently small.  Constraining
$\epsilon$ ensures that the time-scale separation is sufficiently
large. Under similar assumptions, asymptotic stability of the two-time
scale system is guaranteed. As expected, the bound on $\epsilon$ depends on
the selected Lyapunov functions.  For comprehensive discussions on
Lyapunov-based analysis, we refer the reader to~\cite[Ch.7]{PVK-HKK-JOR:99}
and\cite[Ch.~11]{HKK:02}.

% discussion on OFO relation

General results on the closeness of solutions underpin ISS-type results
pervasive in the OFO field~\cite{AH-ZH-SB-GH-FD:24}.  Specifically,
assumptions imposed on the cost functions of the associated optimization
problem, along with the exponential stability of the plant, ensure that the
reduced model (i.e., the feedback controller evaluated at steady state) and
the boundary layer (i.e., the plant) meet the assumptions for the infinite
time interval setting~\cite{PVK-HKK-JOR:99,HKK:02,FCH:66}. Furthermore, the
feedback controller’s gain acts as the time-scale inducing parameter.

To the best of our knowledge, \cite{DDV-JJES:13} and~\cite{DDV-JJES:11} are
the first to present results on the use of nonlinear contraction theory to
analyze singularly perturbed systems. By assuming that the fast and slow
subsystems are contractive, \cite{DDV-JJES:13} provides a bound on the
difference between the trajectories of the two-time scale system and those
of the reduced model.  In addition to the contractivity
assumptions,~\cite{DDV-JJES:11} requires that the largest singular value of
the generalized Jacobian of the two-time scale system be sufficiently
bounded to guarantee contractivity.

% The results in~\cite{DDV-JJES:13,DDV-JJES:11} inspired many research directions, including for the design of disturbance observers for dynamic surface controllers~\cite{rayguru2015contraction}, quantitative bounds for near-decomposability in complex systems~\cite{bousquet2015contraction}, and output feedback controller design~\cite{rayguru2022output}. \red{[maybe just remove last sentence? adds many references.]}

Lyapunov theory for continuous dynamics is not the only approach to analyze
singularly perturbed systems. The ISS properties of two-time scale systems
are investigated in \cite{ART-LM-DN:03}, which establishes ISS properties
by directly inspecting the system’s actual solutions instead of relying on
Lyapunov arguments. Singular perturbation theory has also been applied to
investigate the stability of hybrid systems~\cite{WW-ART-DN:12} and
extremum-seeking schemes~\cite{RJK-WHM-CM:16}.

\textit{Contributions}. % This paper makes three main contributions.
% \noindent \emph{(i)} 

We propose a novel contraction-theoretic approach for the analysis of
general two-time scale nonlinear systems subject to slow and fast
time-varying disturbances, as well as interconnected systems within the OFO
framework.  By assuming that the fast subsystem and the reduced model are
strongly contractive on vector spaces with suitable norms and ensuring that
the time-scale inducing parameter $\eps$ is sufficiently small, we offer
new and distinct results:
\begin{itemize}
\item Theorem~\ref{thm:closeness-of-solns-general}: We guarantee that the
  solutions of the two-time scale system asymptotically approach the
  solutions of its reduced model. Our bounds highlight the roles of
  (one-sided) Lipschitz constants, the time-scale inducing parameter, and
  the time-variability of disturbances.
\item Corollary~\ref{thm:closeness-of-solns-aut}: We derive intuitive
  bounds for the error in the tracking of equilibrium curves, providing
  individual tracking error bounds for the plant state and controller.
      
\item Theorem~\ref{th:ex-LTI-grad-eq-tracking}: We derive new transient and
  asymptotic bounds for the plant and controller separately instead of one
  bound on their concatenated state.
  
\item Theorem~\ref{thm:closeness-of-solns-LTI}: We provide explicit
  transient bounds for the two-time scale LTI solution. In
  Theorem~\ref{thm:str-contractive-whole-system-LTI}, we provide a new
  contractivity result for LTI systems and use weaker assumptions than
  those of the widely used Network Contraction
  Theorem~\cite{GR-MDB-EDS:13},\cite[Thm. 3.23]{FB:24-CTDS}.
\end{itemize}

% Comparison to Slotine paper:

Compared to~\cite{DDV-JJES:13,DDV-JJES:11}, we analyze more general
dynamics under alternative, comprehensive conditions. We also study
important example systems that inspired our work, such as gradient
controllers and LTI systems, which do not satisfy the assumptions
of~\cite{DDV-JJES:13,DDV-JJES:11}. For a detailed comparison
with~\cite{DDV-JJES:13,DDV-JJES:11}, we refer the reader to
Remark~\ref{rem:diff_slotine}.
% \red{[lem 2 -- LTI does not satisfy the uniform bound.]}

% Comparison against Lyapunov:
Separately, our contributions offer several advantages beyond the standard
treatment via Lyapunov theory~\cite{HKK:02}. In
Theorem~\ref{thm:closeness-of-solns-general}, we provide results dependent
on system parameters, yielding closed-form bounds instead of the
$\mathcal{O}(\eps)$ results obtained via Lyapunov
theory~\cite[Thm. 11.2]{HKK:02}.  Moreover, our sufficient conditions on
bounding $\epsilon$ is dependent on identifiable system parameters like
one-sided Lipschitz and Lipschitz constants of the sub-system dynamics
(like the strong convexity constant of a cost function in OFO).  Via
Lyapunov theory, analogous bounds on $\epsilon$ are only shown to
exist~\cite[Sec. 11.1-11.2]{HKK:02} or have expressions that usually depend
on parameters of the Lyapunov functions and are thus not easily
identifiable~\cite[Sec. 11.5]{HKK:02}.

% \red{[should we be comparing this Th 11.2 or an ISS type result since we have disturbances? This comparison is more direct in the simpler cases when we do not have disturbances. I worry here the comparison is more comparable to an ISS type result via Lyapunov theory. But, we do get a $\mathcal{O}(\eps)$ bound, too, so prob okay/nevermind]}. 

Importantly, typical application-specific assumptions that ensure the
dynamics suit Lyapunov-based conditions \emph{also} ensure that the
dynamics satisfy our contractivity assumptions, thus providing stronger
results without additional effort.
For instance, in the case of OFO, adopting these common assumptions allows
us to derive stronger results than those obtained from Lyapunov-based
arguments (see~\cite{LC-GB-EDA:22}). Using these assumptions, we provide
new individual error bounds on the state and controller for equilibrium
tracking to the optimizer. By combining our general theory with that
of~\cite[Thm. 2]{AD-VC-AG-GR-FB:23f}, we achieve greater comprehensibility
of the system’s asymptotic behavior compared to the corresponding
exponential ISS result in~\cite{LC-GB-EDA:22}, which presents one bound on
the concatenated closed-loop state.
Additionally, our bounds explicitly depend on system parameters, such as
the strong convexity constant and the coupling strength of the closed-loop
dynamics, whereas the Lyapunov-based approach leads to a bound also
dependent on the Lyapunov function parameters. Our method also offers
identifiable Lyapunov functions, while the previous analysis relied on the
mere existence of Lyapunov functions satisfying certain bounds.
In summary, our contraction-theoretic approach provides a simpler proof,
more explicit bounds dependent on system parameters, and greater clarity in
the asymptotic behavior of the system compared to traditional
Lyapunov-based methods.

\textit{Organization}.  
We introduce notation and necessary background knowledge in
Section~\ref{sec:notation-and-preliminaries}. Section~\ref{sec:extensions-and-applications}
motivates our study by examining the classic setting of OFO. Our main
results are presented in Section~\ref{sec:main-results}, with detailed
proofs provided in Section~\ref{sec:proofs}. In
Section~\ref{sec:results-for-OFO}, we use our main results to provide novel
asymptotic and transient bounds for the plant and controller separately in
the OFO case. Section~\ref{sec:LTI-dynamics} presents two key results on
the transient bound and overall contractivity of two-time scale LTI
systems. Finally, we conclude the paper in Section~\ref{sec:conclusions}.

%----------------------------------------
\section{Notation and Preliminaries}\label{sec:notation-and-preliminaries}
In this section, we introduce the notation used in the paper and recall basic definitions for contracting dynamical systems. 

\textit{a) Notation}. 
We denote by \(\R, \R_{> 0}, \text{ and } \R_{\geq 0} \) the set of real numbers, the set of positive real numbers, and the set of non-negative real numbers. 
Let $\zero_{n} \in \R^n$ be the vector of all zeros, and $I_n$ be the $n \times n$ identity map. 
For a vector $x \in \R^n$, $x^\top$ denotes transposition and $x_i$ denotes the $i$-th element of $x$. 
For vectors \( x \in \R^n \) and \( u \in \R^m \), \( (x,u) \in \R^{n + m}\) denotes their vector concatenation. We let $\|\cdot\|$ be a norm on $\R^n$ and its corresponding induced norm on $\R^{n \times n}$. Moreover, given  $A \in \R^{n \times n}$, we define the logarithmic norm (log-norm) induced by $\|\cdot \|$ as $\mu(A) := \lim_{h\to 0^+} \frac{\|I_n + h A\| - 1}{h}$. 
Given an invertible matrix $R \in \R^{n \times n}$, matrix norm $\|\cdot\|$, and matrix $A \in \R^{n \times n}$, we define the $R$-weighted norm as $\|A\|_R = \|R A R^{-1}\|$. %Given a symmetric positive definite matrix $P \in \R^{n \times n}$, we let $\| \cdot \|_P$ be defined as $\| \cdot \|_P = \sqrt{x^\top P x}$ for any $x \in \R^n$. 
Given a matrix $A \in \R^{n \times n}$, its spectrum $\operatorname{spec}(A)$ is the set of its eigenvalues, and its spectral abscissa is $\alpha(A) = \max\{ \Re(\lambda) : \lambda \in \operatorname{spec}(A) \}$, with $\Re(s)$ returning the real part of $s \in \mathbb{C}$. 
We let $B_x(y, r):= \{ x \in \R^{n} : \|x - y\| \leq r \} \subset \R^{n}$ denote a ball centered at $y$ with radius $r > 0$. Given a continuous function $g: \R \rightarrow \R$, we denote the upper Dini derivative as $D^+ g$; that is, $D^+ g(x)=\limsup_{h\to {0+}}{\frac {g(x+h)-g(x)}{h}}$. 

Given two normed spaces $(\mathcal{X}, \|\cdot \|_\mathcal{X})$, $(\mathcal{Y}, \|\cdot\|_\mathcal{Y})$, a map $F: \mathcal{X} \to \mathcal{Y}$ is Lipschitz from $(\mathcal{X}, \|\cdot \|_\mathcal{X})$ to $(\mathcal{Y}, \|\cdot\|_\mathcal{Y})$ with constant $\ell \geq 0$ if for all $x_1, x_2 \in \mathcal{X}$, it holds that $\|F(x_1) - F(x_2)\|_\mcY \leq \ell\|x_1 - x_2\|_\mcX$. 
 For maps dependent on multiple parameters, consider the normed spaces $\{ (\mathcal{X}_i, \|\cdot \|_{\mathcal{X}_i}) \}_{i=1}^k,$ $(\mathcal{Y}$ and $ \|\cdot\|_\mathcal{Y})$ and a map $F: \mathcal{X}_1 \times \dots \times \mathcal{X}_k \to \mathcal{Y}$. Without loss of generality, fix $i=1$. The map $F: \mathcal{X}_1 \times \dots \times \mathcal{X}_k \to \mathcal{Y}$.  is Lipschitz from $(\mathcal{X}_1, \|\cdot \|_{\mathcal{X}_1}) \times \dots \times (\mathcal{X}_k, \|\cdot \|_{\mathcal{X}_k}) $ to $(\mathcal{Y}, \|\cdot\|_\mathcal{Y})$ with constant $\ell_{x_1} \geq 0$ if, for all $x, x' \in \mathcal{X}_i$ and for any fixed $z \in \mathcal{X}_2 \times \dots \times \mathcal{X}_k$, it holds that \( \|F(x, z) - F(x', z)\|_\mathcal{Y} \leq \ell_{x_1} \|x - x'\|_{\mathcal{X}_1}\); we say that the map $F: \mathcal{X}_1 \times \dots \times \mathcal{X}_k \to \mathcal{Y}$ is Lipschitz with respect to $x \in \mathcal{X}_1$ with constant $\ell_{x_1} \geq 0$, or notationally, that $\Lip_{x_1} (f) \leq \ell_{f_{x_1}}$ holds uniformly in $z$. Similar definitions hold for $\Lip_{x_i}(f) \leq \ell_{f_{x_i}}$ for $i = 1, 2, \dots, k.$
% \red{[change st the norm and weak pairing is suitably defined on $\R^n.$]} 
 %
A continuous map $F: \R_{\geq 0} \times \R^{n_x} \to \R^{n_x}$ so that $(t,x) \mapsto F(t,x)$ with norm $\|\cdot\|_x$ and compatible weak pairing $\llbracket \cdot ; \cdot \rrbracket_x$ is one-sided Lipschitz with constant $-c \in \R$ if $\llbracket F(t,x) - F(t,x'); x - x' \rrbracket_x \leq -c \|x - x'\|_x^2$ holds for all $t \in \R_{\geq 0}$ and $x,x' \in \R^{n_x}$~\cite[Ch.~2]{FB:24-CTDS}; notationally, we denote the minimum one-sided Lipschitz constant of $F$ as $\osL_x(F) \leq -c.$  %, where $\llbracket \cdot ; \cdot \rrbracket: \R^n \times \R^n \rightarrow \R$ is a  weak pairing on $\R^n$. 
If $F$ is a multi-variable function, the minimum Lipschitz constant of $F$ with respect to variable $(\cdot)$ is denoted by $\Lip_{(\cdot)}(F)$, whereas $\osL_{(\cdot)}(F)$ denotes the minimum one-sided Lipschitz constant.

\textit{b) Contraction Theory.} 
Given a continuous vector field $F: \R_{\geq 0} \times \R^{n_x} \to \R^{n_x}$ so that $(t,x) \to F(t,x)$ with a norm $\|\cdot\|_x$ on $\R^{n_x}$ and compatible weak pairing $\llbracket \cdot ; \cdot \rrbracket_x$ and a constant $c > 0$, we say that $F$ is strongly infinitesimally contracting with respect to $\|\cdot\|_x$ with contraction rate $c$ if, for all $t \geq 0$, the map $x \to F(t,x)$ is one-sided Lipschitz with constant $-c$; i.e., $\osL_x(F) \leq -c$.
For $F$ that is locally Lipschitz in $x$, $\osL_x(F) \leq -c$ if and only if $\mu_x(D F(t,x)) \leq -c$ for all $t \geq 0$ and almost every $x \in \R^{n_x}$, where $D F(t,x) := \partial F(t,x) / \partial x$ is the Jacobian of $F$ with respect to $x$~\cite[Thm.~16]{AD-AVP-FB:22q} (the Jacobian exists for almost every $x$ by Rademacher’s theorem).  
We refer the reader to~\cite{FB:24-CTDS} for a review of the convergence and stability properties of contractive dynamical systems. 

\section{Online Feedback Optimization as a Two-Time Scale System}\label{sec:extensions-and-applications}

% \subsection{Online Feedback Optimization}\label{sec:feedback-optimization}
%

In this section, we motivate our main results by considering a new analysis of Online Feedback Optimization methods. 
% we show how to leverage our results to provide a new analysis of online feedback optimization methods~\cite{MC-ED-AB:20,AH-SB-GH-FD:21} and differential Stackelberg games~\cite{LJR-SAB-SSS:16,fiez2019convergence}.
In the Online Feedback Optimization (OFO) framework, a dynamic plant is interconnected with a controller that is designed based on first-order optimization methods; see the representative works~\cite{MC-ED-AB:20,AH-SB-GH-FD:21}. 
The objective is to regulate inputs and states of the plant to an optimal solution of an optimization problem. For simplicity, we focus on Linear Time-Invariant (LTI) systems and gradient-based controllers here. 
However, our tools are also applicable to nonlinear systems~\cite{AH-SB-GH-FD:21} and projected-gradient controllers~\cite{LC-GB-EDA:22}.  

Consider the following plant, modeled by an LTI system:
\begin{align}\label{eq:LTI-plant}
    \dot z = A z + B u + E w_z, % , ~~~  y = C z + D w
\end{align}
where
$z \in \R^{n_z}$ is the state, 
$u \in \R^{n_u}$ is the control input, 
$w_z \in \R^{w_z}$ is a time-varying unknown disturbance, 
and 
$A, B$, and $E$ are suitably sized matrices.
As is typical, we assume that the matrix $A \in \R^{n_z \times n_z}$ is Hurwitz. 
Then, for any fixed $u$ and $w_z$, the steady state $\subscr{z}{eq}(u,w_z)$ of the LTI system is given by $\subscr{z}{eq}(u,w_z) = G u + H w_z$, with $G:= -A^{-1}B$ and $H := - A^{-1} E$.  
Moreover, we assume that % on par with Assumption~\ref{as:extra-regularity0}, 
$t \to w_z(t)$ is locally absolutely continuous on $\R^{w_z}.$

The  goal of OFO is to design a feedback controller  regulating \eqref{eq:LTI-plant} to solutions of the  optimization problem
\begin{align}\label{eq:opt-prob-1}
    u^*(w_z(t)) \in \operatorname{argmin}_{\bar u \in \R^{n_u}} \phi(\bar u) + \psi(G \bar u + H w_z(t)),
\end{align}
$t \in \R_{\geq 0}$, where $\psi: \R^{n_u} \to \R$ and $\phi: \R^{n_z} \to \R$ are cost functions associated with the system's inputs and outputs, respectively. 

Regarding \eqref{eq:opt-prob-1}, we impose the following regularity assumptions where, for simplicity, $\|\cdot\|$ denotes the $\ell_2$-norm.

\vspace{.1cm}

\begin{assumption}\label{A:Lipschitz-costs}
The functions $u \to \phi(u)$ and $z \to \psi(z)$ are continuously differentiable. 
There exist constants $\ell_\phi, \ell_\psi \geq 0$ such that, for every $u, u' \in \R^{n_u}$ and $z, z' \in \R^{n_z}$, respectively, the following holds: $\|\nabla \phi(u) - \nabla \phi(u')\| \leq \ell_\phi \|u - u'\|$ and $\| \nabla \psi(z) - \nabla \psi(z')\| \leq \ell_\psi \|z - z'\|$.
\hfill $\triangle$
\end{assumption}

\vspace{.1cm}

\begin{assumption}\label{A:strongly-convex-cost}
The function $u \to \phi(u)$ is strongly convex with parameter $\nu > 0$; function $z \to \psi(z)$ is convex. \hfill $\triangle$ 
\end{assumption}

\vspace{.1cm}

These assumptions imply that, for any fixed $w_z$, the function $u \to \phi(u) +  \psi(G u + H w_z)$ is $\ell$-smooth with $\ell := \ell_\phi + \|G\|^2 \ell_\psi$ and  strongly convex with parameter $\nu > 0$ in the $\ell_2$-norm. 

In practice, such an optimization problem with Assumptions~\ref{A:Lipschitz-costs}-\ref{A:strongly-convex-cost} can represent  % \textcolor{red}{add comment on example here: power systems, traffic flow control in transportation networks, control of epidemics, neuroscience, autonomous driving, robotics (see references in OJCSYS intro).}
optimization problems in power systems~\cite{MC-ED-AB:20,AH-SB-GH-FD:21,SM-AH-SB-GH-FD:18,LSPL-ZEN-EM-JWSP:18} 
and neuroscience~\cite{yang2021modelling}, traffic flow control in transportation networks~\cite{GB-JC-JIP-EDA:22}, and epidemic control~\cite{bianchin2021planning}. 
When disturbances are not present in the plant model, this setting % the problem~\eqref{eq:opt-prob-1} 
can model problems in
% when the dynamical model for the plant does not include exogenous disturbances, our optimization-based controllers can also be utilized in the context of 
autonomous driving~\cite{dean2020robust,dean2021certainty} and robotics~\cite{AT-SF-MP-MHdB-FD:22}.

Altogether, the open-loop gradient flow reads
\begin{align}
    \dot u = - \eps\left( \nabla \phi(u) + G^\top \nabla \psi(G u + H w_z)\right) \label{eq:ex-LTI-openloop},
\end{align}
where $\eps > 0$ is a small controller gain to be tuned,
can be shown to be strongly infinitesimally contractive with contraction rate $\nu$ for each fixed $w_z$; see~\cite{AD-VC-AG-GR-FB:23f} for details. 

%Later, we confirm that the open-loop gradient flow~\eqref{eq:ex-LTI-openloop} is, in fact, our reduced system and, thus, satisfies Assumption~\ref{as:red-str-contracting}. 

Given the problem~\eqref{eq:opt-prob-1}, the OFO design strategy involves using the gradient flow~\eqref{eq:ex-LTI-openloop} as a state-feedback controller:
\begin{subequations}
\label{eq:ex-LTI-grad-temp}
    \begin{align} 
        \dot u &= - \eps \left( \nabla \phi(u) + G^\top \nabla \psi(z) \right), \,\, ~~ u(0) = u_0,\label{eq:ex-LTI-grad-b-temp} \\
        \dot z &= A z + B u + E w_z, ~~~~~~~~~~~~~~ z(0) = z_0 .\label{eq:ex-LTI-grad-a-temp}
    \end{align}
\end{subequations}
By comparing~\eqref{eq:ex-LTI-openloop} with~\eqref{eq:ex-LTI-grad-b-temp}, one can notice that measurements of the state $z$ are used in the gradient flow instead of the steady state $z_{\text{eq}}(u,w_z)$. 

Also, the controller gain $\eps > 0$ induces a time-scale separator between the plant and controller. As such, we may leverage singular perturbation arguments to analyze the system~\eqref{eq:ex-LTI-grad-temp}. To match conventional notation used in singular perturbation theory, we will equivalently rewrite % reorganize 
the system as follows:
\begin{subequations}
\label{eq:ex-LTI-grad}
    \begin{align} 
        \dot u &= - \nabla \phi(u) - G^\top \nabla \psi(z) , \,\, ~~ u(0) = u_0,\label{eq:ex-LTI-grad-b} \\
        \eps \dot z &= A z + B u + E w_z, ~~~~~~~~~ z(0) = z_0 ,\label{eq:ex-LTI-grad-a}
    \end{align}
\end{subequations}
where $\eps > 0$ more clearly induces a time-scale separation between the plant and the controller. % \textcolor{red}{LC: I think this is all equivalent, right (moving the gain over to $\dot z$)? Want to make sure the controller set up matches singular perturbation notation.}

Next, we will consider and provide results for a more general set of dynamics in Section~\ref{sec:main-results}, with proofs in Section~\ref{sec:proofs}, before returning to analyze the asymptotic behavior of system~\eqref{eq:ex-LTI-grad} in Section~\ref{sec:results-for-OFO}.

\section{Singular Perturbation via Contraction Theory}\label{sec:main-results}

\subsection{General Singular Perturbation Framework}\label{sec:general-nonlinear-dynamics-with-disturbances}

This section presents our main results for the analysis of two-time scale systems using contraction theory. 
First, we present a general analysis for nonlinear, singularly perturbed systems with time-varying exogenous inputs in Section~\ref{sec:main-results}. 
%The analysis is tailored to autonomous dynamical systems in Section~\ref{sec:simplified-nonlinear-dynamics}. 
%Finally, new results for the interconnection of LTI systems are presented in Section~\ref{sec:LTI-dynamics}. 

We consider the following \textit{two-time scale system}: 
\begin{subequations}
\label{eq:gen-system}
    \begin{align}\label{eq:gen-system-a}
    \dot x &= f(t,x,z,w_x,\eps), ~~ x(0) = x_0\\ 
    \label{eq:gen-system-b}
    \eps \dot z &= g(x,z,w_z,\eps), ~~~~~ z(0) = z_0,
\end{align}
\end{subequations}
where $t \geq 0$ denotes time, $\epsilon \geq 0$ is the time-scale inducing parameter, 
$x \in \R^{n_x}$ and  $z \in \R^{n_z}$ are the states of the two sub-systems  (with $x_0 \in \R^{n_x}$ and $z_0 \in \R^{n_z}$ as initial conditions), and $w_x: \R_{\geq 0} \to \R^{w_x}$,
$w_z: \R_{\geq 0} \to \R^{w_z}$ denote time-varying exogenous inputs. The vector fields $f: \R_{\geq 0} \times \R^{n_x} \times \R^{n_z} \times \R^{w_x} \times \R_{\geq 0} \to \R^{n_x}$ and $g: \R^{n_x} \times \R^{n_z} \times \R^{w_z} \times \R_{\geq 0} \to \R^{n_z}$ are continuous 
in their arguments.
Finally, assume there exist norms $\|\cdot\|_{x}, \|\cdot \|_{z}, \|\cdot\|_{w_x},$ and $ \|\cdot\|_{w_z}$  on $\R^{n_x}, \R^{n_z}, \R^{w_x},$ and $\R^{w_z}$, respectively, that are possibly different, and ensure such that $\|\cdot\|_x$ and $\|\cdot\|_z$ have compatible weak pairings $\llbracket \cdot ;  \cdot \rrbracket_x$ and $\llbracket \cdot ; \cdot \rrbracket_z$.
The parameter $\epsilon$ (which is taken to be $\eps < 1$ in standard singular perturbation arguments) induces a timescale separation between the two sub-systems, where~\eqref{eq:gen-system-a} is the  \emph{slow sub-system} and~\eqref{eq:gen-system-b} is the \emph{fast sub-system}~\cite{PVK-HKK-JOR:99,HKK:02}. By setting $\eps = 0$, one can eliminate the derivatives of $z$ and, thus, reduce the order of the model. The system~\eqref{eq:gen-system} is in \textit{standard form} when, for $\eps = 0$,~\eqref{eq:gen-system-b} has isolated real roots. 
For the two-time scale system~\eqref{eq:gen-system}, we make the following assumptions.

% \red{[Careful: $\osL$ is defined using weak pairing, so for $\|\cdot\|_x$ and $\|\cdot\|_z,$ we need the weak pairings to show up here. Change notation.]}

\begin{assumption}[Contractivity of the fast dynamics]
\label{as:fast-str-contracting}
    There exists $c_g > 0$ such that $\osL_z(g(x,z,w_z,\eps)) \leq -c_g$ holds uniformly in $x$, $w_z$, and $\eps$. % \blue{with respect to $\|\cdot\|_{z}$}.
    \hfill $\triangle$
\end{assumption}

\begin{assumption}[Continuity of slow dynamics]
\label{as:Lipschitz-f-in-t}
  There exists $\ell_{f_x} \geq 0$ such that $\Lip_x (f(t,x,z,w_x,\eps)) \leq \ell_{f_x}$ holds uniformly in $z$, $w_x$, $\eps$, and $t$.
    \hfill $\triangle$
\end{assumption}

\begin{assumption}[Lipschitz interconnections]
\label{as:Lipschitz-interconnection-x-z}
    There exists $\ell_{f_z} > 0$ such that $\Lip_z (f(t,x,z,w_x,\eps)) \leq \ell_{f_z}$ holds uniformly in $x$, $w_x$, $\eps$, and $t$. 
    There exists $\ell_{g_x} > 0$ such that $\Lip_x (g(x,z,w_z,\eps)) \leq \ell_{g_x}$ holds uniformly in $z$, $w_z$, $\eps$, and $t$. 
    \hfill $\triangle$
\end{assumption}

% pushed to new line bc the \ells \geq 0 ran over the edge of the page
\begin{assumption}[Continuity in the parameters]
\label{as:Lipschitz-interconnection-else}

\noindent There exists $\ell_{f_t}, \ell_{f_w}, \ell_{f_\eps}, \ell_{g_w}, \ell_{g_\eps} \geq 0$ such that:

\noindent \ref{as:Lipschitz-interconnection-else}.\text{i)} \hspace{.1cm} $\Lip_t (f(t,x,z,w_x,\eps)) \leq \ell_{f_t}$ uniformly in $x,z,w_x, \eps$;

\noindent \ref{as:Lipschitz-interconnection-else}.\text{ii)} \hspace{.0cm} $\Lip_{w_x} (f(t,x,z,w_x,\eps)) \leq \ell_{f_w}$ uniformly in $x,z,\eps, t$; 

\noindent \ref{as:Lipschitz-interconnection-else}.\text{iii)} \hspace{-.07cm} $\Lip_{\eps} (f(t, x, z, w_x, \eps)) \leq \ell_{f_\eps}$ uniformly in $x,z,w_x, t$;   

\noindent \ref{as:Lipschitz-interconnection-else}.\text{iv)} \hspace{-.025cm} $\Lip_{w_z} (g(x, z, w_z, \eps)) \leq \ell_{g_w}$ uniformly in $x,z,\eps,  t$; 
    
\noindent \ref{as:Lipschitz-interconnection-else}.\text{v)} \hspace{.052cm} $\Lip_{\eps} (g(x,z,w_z, \eps)) \leq \ell_{g_\eps}$ uniformly in $x,z,w_z, t$.
\hfill $\triangle$
\end{assumption}

\begin{assumption}[Smoothness of the disturbances]
\label{as:extra-regularity0}
    The maps $t \to w_z(t)$ and $t \to w_x(t)$ are locally absolutely continuous on $\R^{w_z}$ and on  $\R^{w_x}$, respectively.  \hfill $\triangle$
\end{assumption}

Since the vector field $g$ is continuous and satisfies Assumption~\ref{as:fast-str-contracting}, then a solution $z(t)$ of~\eqref{eq:gen-system-a} exists and is unique \cite[Exercise E3.4]{FB:24-CTDS}. For the slow system~\eqref{eq:gen-system-b}, a solution $x(t)$ is guaranteed to exist and be unique since $f$ is continuous and satisfies Assumption~\ref{as:Lipschitz-f-in-t}% by the Picard-Lindelof Theorem
~\cite[E1.2]{FB:24-CTDS}.

\begin{remark}{(Local vs global sets):}\label{rem:localvsglobalsets}
    Assumptions~\ref{as:fast-str-contracting}-\ref{as:extra-regularity0} may hold locally on convex, forward-invariant sets; we choose to write these assumptions globally for simplicity of exposition. Further, the OFO setting satisfies these assumptions globally. % \red{comment on OFO satisfies this! [easier to extend to case of local contractivity over convex, forward-invariant sets at the expense of heavier notation]}
\end{remark}

Assumption~\ref{as:fast-str-contracting} implies that, for given $x, w_z$ and $\eps$, the equation 
$g(x,z,w_z,\eps) = 0$ has a unique solution~\cite[Thm.~3.9]{FB:24-CTDS}, which is hereafter denoted as $z^*(x, w_z,\eps)$. 
With $\eps = 0$,  
$z^*(x, w_z,0)$ is the root of $g(x,z,w_z,0) = 0$~\cite{PVK-HKK-JOR:99}; to simplify notation,  we define $z^*(x, w_z) := z^*(x, w_z,0)$. Assumption~\ref{as:fast-str-contracting} also implies that, for any $\eps > 0$, the fast sub-system~
\eqref{eq:gen-system-b} is strongly infinitesimally contractive with rate $c_g/\eps$ in the domain $t$.  
Assumption~\ref{as:extra-regularity0} ensures that the essential suprema of $\|\dot w_z\|$ and of $\|\dot w_x\|$ are well defined. %, which will be used in our main results. 

Lastly, we note the following properties and form an assumption on $z^*(x,w_z,\eps)$.
% Our proofs will leverage the following lemma about the unique solution to $g(x,z,w_z,\eps) = 0$. 

\begin{lemma}[Lipschitz solutions to $g(x,z,w_z,\eps) = 0$]
\label{lem:lipz}
Let Assumptions~\ref{as:fast-str-contracting} and~\ref{as:Lipschitz-interconnection-else} hold, and let $z^*(x, w_z,\eps)$ be the unique solution to 
    $g(x,z,w_z,\eps) = 0$ for given $x, w_z$ and $\eps$. Then, the following holds: 

\noindent \emph{i)} For any fixed $w_z$ and $\eps$, the map $x \mapsto z^*(x,w_z,\eps)$ is Lipschitz with constant $\frac{\ell_{g_x}}{c_g}$. 

\noindent \emph{ii)} For any fixed $x$ and $\eps$, the map $w_z \mapsto z^*(x, w_z,\eps)$ is Lipschitz with constant $\frac{\ell_{g_w}}{c_g}$. 

\noindent \emph{iii)} For any fixed $x$ and $w_z$, the map $\eps \mapsto z^*(x, w_z,\eps)$ is Lipschitz with constant $\frac{\ell_{g_\eps}}{c_g}$. 
\QEDB
\end{lemma}

% Citing Sasha's papers: 
% https://arxiv.org/pdf/2305.15595
% https://arxiv.org/pdf/2110.08298
The Lipschitz constants in Lemma~\ref{lem:lipz} can be derived using the arguments of~\cite[Lem. 1]{AD-VC-AG-GR-FB:23f}. 
Further, the Rademacher Theorem states that if a function is Lipschitz continuous on an open subset of Euclidean space, then it is almost everywhere differentiable. Hence, item $i)$ of Lemma~\ref{lem:lipz} means that for every $t \geq 0$ and almost every $x \in \R^{n_x}$, the map $x \mapsto z^*(x, w_z, \eps)$ is almost everywhere differentiable for given $w_z $ and $ \eps$ (and so on for items $ii)$ and $iii)$).

%%%%%%%% REMOVING BC WE DON'T NEED ANYMORE: 
% \red{We don't need this assumption... We needed it in the old local contraction theorem (which we don't have anymore). Right?}
% \begin{assumption}[Differentiability of the fixed point map $z^*$]
% \label{as:extra-regularity} 
% %For any fixed $w_z$ and $\eps$, the map  $x \mapsto z^*(x, w_z,\eps)$ is globally differentiable; for any fixed $x$ and $\eps$, $w_z \mapsto z^*(x, w_z,\eps)$ is globally differentiable. %Moreover,  
% There exists $\ell_{z^*_w} \geq 0$ such that
% $\Lip_{x} (\frac{\partial z^*}{\partial w_z}(x,w_z,\eps)) \leq \ell_{z^*_w}$ holds uniformly in $w_z$ and $\eps$. \hfill $\triangle$
% % , for any $w_z \in \R^{w_z}$, $\eps \geq 0$, the map $x \mapsto \frac{\partial z^*}{\partial w_z}(x,w_z,\eps)$ is Lipschitz with constant $\ell_{z^*_w}.$ \hfill $\triangle$
% \end{assumption}

% \vspace{.1cm} 

% \textcolor{red}{Do we remove the first two sentences since we have Lemma~\ref{lem:lipz}, and by Rademacher theorem, we get that the map $z^*$ is almost everywhere differentiable? Keep the last sentence of As 8?}

% Assumption~\ref{as:extra-regularity} is a reasonable ask, since by t

Following the standard singular perturbation framework~\cite{PVK-HKK-JOR:99,HKK:02}, we now define the reduced model and the boundary layer system associated with the two-time scale system~\eqref{eq:gen-system}. The \textit{reduced model} is defined as:
    \begin{align}\label{eq:gen-reduced}
        \dot x_{\text{r}} & = f_\text{red}(t,x_{\text{r}}, w_x,w_z)  := f(t,x_{\text{r}}, z^*(x_{\text{r}},w_z), w_x,0) 
    \end{align}   
with initial condition $x_{\text{r}}(0) = x_{\text{r},0}$, where the fast state variable $z$ is substituted by its quasi-steady state $z^*(x_{\text{r}},w_z)$. 
Since the reduced model~\eqref{eq:gen-reduced} is fictitious and only used to compare against the slow dynamics~\eqref{eq:gen-system-a}, 
% \red{[change language from customary]} 
we select $x_r(0) = x(0)$ without loss of generality; this will simplify some terms in our analysis. 
Effectively, the state $x_{\text{r}}(t)$ of the reduced model \eqref{eq:gen-reduced} serves as an approximation of the state $x(t)$ of the original slow sub-system~\eqref{eq:gen-system-a}. 
On the other hand, $z^*(x_{\text{r}},w_z)$ is expected to be an approximation of the fast state variable $z(t)$ only after an initial time interval,
so a correction term must be utilized to compensate for $z_0 - z^*(x_{\text{r},0}, w_z(0))$. 
The dynamics for this initial adjustment is defined as the \textit{boundary layer system}, given by:
    \begin{align}
    \label{eq:gen-BL}
        \begin{split}
        \frac{d y_{\text{bl}}}{d \tau} 
        = g(x, y_{\text{bl}} + z^*(x,w_z),w_z, 0), 
        ~~~~\tau = \frac{t}{\eps},
        \end{split}
    \end{align}
where $\frac{d y_{\text{bl}}}{d \tau} = \eps \frac{d y_{\text{bl}}}{dt}$, 
$\tau = 0$ is the initial value when $t = 0$, 
and $x$ is treated as a slowly varying parameter~\cite{PVK-HKK-JOR:99,HKK:02}. 
While~\eqref{eq:gen-BL} is contractive with rate $c_g$ in the time scale $\tau$, we make the following assumptions on the reduced model. 

\vspace{.1cm}

\begin{assumption}[Contractivity of the reduced model]
\label{as:red-str-contracting}
    There exists $c_f > 0$ such that $\osL_x(f_\text{red}(t,x_{\text{r}}, w_x,w_z)) \leq -c_f$ holds uniformly in $w_x$, $w_z$, and $t$. \hfill $\triangle$ 
\end{assumption}

% \begin{table}[t!]
%     \centering
%     \begin{tabular}{p{3.3cm}|p{3.3cm}} 
%  \hline \hline
% { \vspace{-.4cm}
% \begin{align*}
%     c_f &= -\osL_x(f_\text{red}) \\
%  \ell_{f_z} &= \Lip_z(f)  \\
%  \ell_{f_\eps} &= \Lip_\eps(f) \\
%  \ell_{g_w} &= \Lip_{w_z} (g)  \\
% \ell_{z^*_w} &= \Lip_x \left(\partial z^* / \partial w_z\right)  \\
%  \ell_{f_t} &= \Lip_t(f) \\
%  L_{\eps,f_x} &= \Lip_x(D_\eps(f)) \\
% \ell_{f_x} &= \Lip_x(f)
% \end{align*}\vspace{-.4cm}} & {\vspace{-.4cm} \begin{align*}
%     c_g &= -\osL_z(g) \\
%  \ell_{g_x} &= \Lip_x(g)  \\
%  \ell_{g_\eps} &= \Lip_\eps (g) \\
%  \ell_{f_{w_x}} &= \Lip_{w_x} (f)  \\
%  \ell_{f_{w_z}} &= \Lip_{w_z}(f_{\text{red}}) \\
%  L_{z,f_x} &= \Lip_x(D_z(f)) \\
% L_{z,g_x} &= \Lip_x(D_z(g)) \\
% L_{\eps, g_x} &= \Lip_x(D_\eps(g)) 
% \end{align*} \vspace{-.4cm}}  \\
%  \hline \hline
%     \end{tabular}
%     \caption{Summary of the constants in Assumptions~\ref{as:fast-str-contracting}-\ref{as:red-str-contracting} \vspace{-0.5cm}}
%     \label{tab:table-of-Lips-symbols}
% \end{table}

\begin{table}[t!]
    \centering
    \begin{tabular}{p{3.3cm}|p{3.3cm}|p{3.3cm}} 
 \hline \hline
{ \vspace{-.4cm}
\begin{align*}
    c_f &= -\osL_x(f_\text{red}) \\
 \ell_{f_z} &= \Lip_z(f)  \\
 \ell_{f_\eps} &= \Lip_\eps(f) \\
 \ell_{g_w} &= \Lip_{w_z} (g) % \\
% \ell_{f_t} &= \Lip_t(f) \\
% L_{\eps,f_x} &= \Lip_x(D_\eps(f)) \\
%\ell_{f_x} &= \Lip_x(f)
\end{align*}\vspace{-.4cm}} & { \vspace{-.4cm}
\begin{align*}
%     c_f &= -\osL_x(f_\text{red}) \\
%  \ell_{f_z} &= \Lip_z(f)  \\
%  \ell_{f_\eps} &= \Lip_\eps(f) \\
%  \ell_{g_w} &= \Lip_{w_z} (g)  \\
% \ell_{z^*_w} &= \Lip_x \left(\partial z^* / \partial w_z\right)  \\
\ell_{z^*_w} &= \Lip_x \left(\partial z^* / \partial w_z\right)  \\
 \ell_{f_t} &= \Lip_t(f) \\
% L_{\eps,f_x} &= \Lip_x(D_\eps(f)) \\
\ell_{f_x} &= \Lip_x(f) \\
 \ell_{f_{w_z}} &= \Lip_{w_z}(f_{\text{red}}) 
%L_{z,g_x} &= \Lip_x(D_z(g)) \\
%L_{\eps, g_x} &= \Lip_x(D_\eps(g) \\
~&~
\end{align*}\vspace{-.4cm}} & {\vspace{-.4cm} \begin{align*}
    c_g &= -\osL_z(g) \\
 \ell_{g_x} &= \Lip_x(g)  \\
 \ell_{g_\eps} &= \Lip_\eps (g) \\
 \ell_{f_{w_x}} &= \Lip_{w_x} (f)  
 %L_{z,f_x} &= \Lip_x(D_z(f)) \\
\end{align*} \vspace{-.4cm}}  \\
 \hline \hline
    \end{tabular}
    \caption{Summary of the constants in Assumptions~\ref{as:fast-str-contracting}-\ref{as:red-str-contracting} \vspace{-0.5cm}}
    \label{tab:table-of-Lips-symbols}
\end{table}

Hereafter, we denote $t \mapsto
\subscr{x}{r}(t)$ as a solution to the reduced
model~\eqref{eq:gen-reduced}, which is guaranteed to exist and be unique~\cite[Exercise E3.4]{FB:24-CTDS}. 
After presenting our main results, we will discuss how our assumptions compare with the classical analysis using Lyapunov theory~\cite{PVK-HKK-JOR:99,HKK:02,FCH:66} in Remark~\ref{rem:comp_Lyap}.

To state some of our results, it is convenient to consider the
two-time scale system~\eqref{eq:gen-system} after a variable shift. For
fixed $w_z$ and $\eps$, consider the shifted fast variable $y := z -
z^*(x,w_z)$; then, the two-time scale system~\eqref{eq:gen-system} in the new variables defines the \textit{shifted system}, given as:
\begin{subequations}
\label{eq:gen-new-var-w}
    \begin{align}\label{eq:gen-new-var-w-a}
        \dot x &= f(t,x, y + z^*(x,w_z), w_x,\eps) \\
 %           \eps \dot y &= g(x, y + z^*(x,w_z), w_z,\eps) - \eps \frac{\partial z^*}{\partial w_z}(x,w_z) \dot w_z \nonumber \\
 %           & \quad\quad\quad\quad\quad\quad\quad\quad\quad\quad\quad\quad\quad\quad - \eps \frac{\partial z^*}{\partial x}(x, w_z) f(t, x, y + z^*(x,w_z), w_x,\eps) \label{eq:gen-new-var-w-b}            \\
            \eps \dot y &=  g(x, y + z^*(x,w_z), w_z,\eps) - \eps \frac{\partial z^*}{\partial w_z}(x,w_z) \big( \dot w_z - f(t, x, y + z^*(x,w_z), w_x,\eps)  \big)  \label{eq:gen-new-var-w-b}
    \end{align}
\end{subequations}
% {\small\color{red} if we write $\dot y$ in this way, we can condense it to one line!} AGREED!
% {\color{blue} From FB: I prefer it this way; if you agree, please erase this blue words}
with initial conditions $x(0) = x_0$ and $y(0) = z_0 - z^*(x_0, w_z(0))$. 
The reduced model~\eqref{eq:gen-reduced} and the boundary layer
system~\eqref{eq:gen-BL} can be seen as the limiting
systems~\cite{HKK:02} of the shifted system~\eqref{eq:gen-new-var-w} for $\eps \rightarrow
0^+$.

Our main results for the two-time scale system~\eqref{eq:gen-system} are presented next. The first result pertains to the closeness of the trajectories of the two-time scale system~\eqref{eq:gen-system} and of the pair $\subscr{x}{r}(t)$, $z^*(x_\text{r}(t),w_z(t))$. 
For ease of reference, Table~\ref{tab:table-of-Lips-symbols} lists the relevant Lipschitz constants. 

\vspace{.1cm}

\begin{theorem}[Closeness of the solutions]
\label{thm:closeness-of-solns-general}
Consider the two-time scale system \eqref{eq:gen-system} satisfying the  (contractivity) Assumptions~\ref{as:fast-str-contracting} and~\ref{as:red-str-contracting}, (continuity) Assumptions~\ref{as:Lipschitz-f-in-t}-\ref{as:Lipschitz-interconnection-else},  
and (disturbance) Assumption~\ref{as:extra-regularity0}. 
%%%%%%%% REMOVING BC WE DON'T NEED ANYMORE: 
%, and (differentiability of $z^*$) Assumption~\ref{as:extra-regularity}. 
Suppose that,  
\begin{align}
   0< \eps < \eps_0^* := \frac{c_g^2}{\ell_{g_x} \ell_{f_z}}.
\end{align} 
Then, for any initial condition, the two-time scale system~\eqref{eq:gen-system} has a
unique solution $x(t), z(t)$, for $t \geq 0$ satisfying
\allowdisplaybreaks
\begin{subequations}\label{eq:thm-close-solns-general}
\begin{align}
       \|x(t) - \subscr{x}{r}(t)\|_{x} &\leq \frac{\eps }{c_g - \eps \frac{\ell_{g_x}}{c_g} \ell_{f_z}} ( \delta_{x,2} + \delta_{x,3} \bar{w}_z + \delta_{x,4} \bar{w}_x)  \nonumber \\ 
       & \qquad \quad + \frac{\eps^2 }{c_g - \eps \frac{\ell_{g_x}}{c_g} \ell_{f_z} } \delta_{x,1} + E_x(t,\eps)   \label{eq:thm-close-solns-general-a}, \\
       \|z(t)  - z^*(x_\text{r}(t),w_z(t))\|_{z} &\leq  \frac{\eps^2 }{c_g - \eps \frac{\ell_{g_x}}{c_g} \ell_{f_z}} \delta_{z,1} \nonumber \\ 
       & \qquad \quad + \frac{\eps }{c_g - \eps \frac{\ell_{g_x}}{c_g} \ell_{f_z}}( \delta_{z,2} + \delta_{z,3} \bar{w}_z + \delta_{z,4} \bar{w}_x)  + E_z(t, \eps),
     \label{eq:thm-close-solns-general-b} 
\end{align}
\end{subequations}
where $t \mapsto \subscr{x}{r}(t)$ is a solution of the reduced
model~\eqref{eq:gen-reduced}, $\bar w_x := \operatorname{ess~sup}_{s \geq
  0} \|\dot w_x(s)\|_{w_x}$, $\bar w_z := \operatorname{ess~sup}_{s \geq 0}
\|\dot w_z(s)\|_{w_z}$, the maps $t \mapsto E_x(t,\eps)$ and $t \mapsto
E_z(t,\eps)$ tend to $0$ as $t \rightarrow +\infty$, and the positive
constants $\{\delta_{x,i}, \delta_{z,i}\}_{i = 1}^4$ are defined by
\begin{align}\label{eq:thm1-deltas}
  \begin{split}
    \delta_{x,1} &:=  \frac{\ell_{f_z} \ell_{g_x} \ell_{f_\eps}}{c_g}, ~ \delta_{x,2} :=  \ell_{f_z} \ell_{g_\eps} + \frac{\ell_{f_z} \ell_{g_x} \ell_{f_t}}{c_f c_g} + \ell_{f_\eps} c_g , ~ \delta_{x,3} := \frac{\ell_{f_z} \ell_{g_w}}{c_g} + \frac{\ell_{f_z}^2 \ell_{g_x} \ell_{g_w}}{c_f c_g^2} , \\
    \delta_{x,4} &:= \frac{\ell_{f_z} \ell_{g_x} \ell_{f_w}}{c_f c_g} , ~ \delta_{z,1} := 
    \frac{ c_f \ell_{f_\eps}}{\ell_{f_w} }\delta_{z,4} , ~ \delta_{z,2} := \left(1 + \frac{\ell_{g_x} \ell_{f_z}}{c_g}\right)\left(\ell_{g_\eps} + \frac{\ell_{g_x} \ell_{f_t}}{c_f c_g}\right) + \ell_{g_x} \ell_{f_\eps} , \\
    \delta_{z,3} &:=   \left(\frac{\ell_{g_w}}{c_g} + \frac{\ell_{g_x} \ell_{g_w}}{c_f c_g^2} \right)\left(1 + \frac{\ell_{g_x} \ell_{f_z}}{c_g}\right), ~\delta_{z,4} :=  \frac{\ell_{g_x} \ell_{f_w}}{c_g c_f}\left(1 + \frac{\ell_{g_x} \ell_{f_z}}{c_g}\right) .
    \qquad\qquad\qquad \QEDB
  \end{split}
\end{align}
\end{theorem}

The proofs of the two theorems are presented in  Section~\ref{sec:proofs}. In the proof for Theorem~\ref{thm:closeness-of-solns-general}, we provide explicit expressions for the functions $E_x(t,\eps)$ and $E_z(t,\eps)$ (which are not reported in the theorem for simplicity of exposition).

Theorem~\ref{thm:closeness-of-solns-general} establishes that  the trajectory $x(t)$ of the slow sub-system~\eqref{eq:gen-system-a}  approaches the solution $\subscr{x}{r}(t)$ of the reduced model~\eqref{eq:gen-reduced} and then remains close to $\subscr{x}{r}(t)$ within a ball proportional to $\eps$. 
Similarly, the trajectory $z(t)$ of the fast sub-system~\eqref{eq:gen-system-b} approaches its quasi-steady-state $z^*(\subscr{x}{r}(t), w_z(t))$ and  
then remains within a ball  proportional to $\eps$. In the two bounds of~\eqref{eq:thm-close-solns-general}, the effects of the disturbances are identifiable.
The results of Theorem~\ref{thm:closeness-of-solns-general} are the counterpart of~\cite[Thm.~11.2]{HKK:02}, in which the closeness of the trajectories of the shifted system~\eqref{eq:gen-new-var-w} and the limiting system  is analyzed by leveraging  Lyapunov theory; see also~\cite{kokotovic1976singular} and~\cite{FCH:66}. 

% Theorem~\ref{thm:str-contractive-whole-system-general} provides sufficient
% conditions for the shifted system~\eqref{eq:gen-new-var-w} to be strongly
% infinitesimally contracting and specifies its contraction rate. Contractivity is obtained for sufficiently
% small values of $\eps$, depending on the Lipschitz constants of the vector
% fields and the time-variability of the disturbances.

We now provide a few remarks and comments.

\vspace{.1cm}

\begin{remark}[Comparison with Lyapunov stability analysis]
\label{rem:comp_Lyap}
Classical results on the infinite time interval based on  Lyapunov theory rely  on the assumption that the reduced model and the  boundary layer system have equilibria that are (locally) exponentially stable; see, for example, ~\cite[Thm.~11.2]{HKK:02} and~\cite{FCH:66}. Similar assumptions are used to establish exponential stability of equilibria of the singularly perturbed system as  in~\cite[Thm.~11.4]{HKK:02}. 
In our analysis, we replace the exponential stability assumptions with contractivity properties
on the fast dynamics and reduced model via Assumptions~\ref{as:fast-str-contracting} and~\ref{as:red-str-contracting}, respectively. 
We also provide explicit bounds on the time-scale inducing parameter $\eps$, whereas the Lyapunov-based analysis yields only an existence of a bound on $\eps$.
Our assumptions allow us to directly compare
the fast trajectory $z(t)$ with its steady state $z^*(x_\text{r}(t), w_z(t))$ for all $t \geq 0$ in Theorem~\ref{thm:closeness-of-solns-general} instead 
of only after an initial time interval as in~\cite{FCH:66,HKK:02}. 

%Regarding Theorem~\ref{thm:str-contractive-whole-system-general}, our bounds on $\eps$ can be computed from known system properties.
%In contrast, the bounds on $\eps$ via Lyapunov analysis depends on parameters of the selected Lyapunov functions associated with the reduced model and the boundary layer system~\cite[Ch.~7]{PVK-HKK-JOR:99}. 
%
Finally, it is important to note a limitation in our model: 
our results are applicable to two-time scale systems where the vector field $g(x,z,w_z,\eps)$ is not dependent on time. While this is aligned with~\cite{DDV-JJES:13}, the Lyapunov-based analysis can handle time-varying fast sub-systems~\cite[Ch.~7]{PVK-HKK-JOR:99},~\cite[Ch.~11]{HKK:02}. 
\hfill $\triangle$
\end{remark}

\vspace{.1cm}

\begin{remark}[Comparison with two-time scale contractivity] %~\cite{DDV-JJES:13,DDV-JJES:11}]
\label{rem:diff_slotine}
     Our contribution relative to the seminal works~\cite{DDV-JJES:13,DDV-JJES:11} is threefold. 
    First, our analysis considers dynamics of the form of the two-time scale system~\eqref{eq:gen-system}, where  disturbances $t \mapsto w_x(t)$ and $t \mapsto w_z(t)$ enter the slow and fast sub-systems, respectively, and the parameter $\eps$ is an argument of both of the dynamics in the two-time scale system~\eqref{eq:gen-system}. 
    On the other hand, the system in~\cite{DDV-JJES:13} is of the form $\dot x = f(t, x, z)$, $\eps \dot z = g(x, z, \eps)$.

    % \red{[assumption in lem 2 is not satisfied for LTI. moreover, the assumptions of th 3 show that this assumption is not required due to additionally assumign that f is contractive.]}
    Second, the results pertaining to the closeness of solutions in~\cite{DDV-JJES:13} are derived under stronger assumptions than ours. 
    In Lemma~2 of \cite{DDV-JJES:13} the assumptions are: (a1) the fast sub-system is contractive, and (a2) the quantity $\|\frac{\partial z^*}{\partial x} f(t,x,z,w_x,0)\|_x$ is uniformly bounded in all of its arguments. 
    For our computations, we do not need (a2) to compute a comparable bound, which is given as an intermediary result for the proof of Theorem~\ref{thm:closeness-of-solns-general} in Lemma~\ref{lem:bound-y} of Section~\ref{sec:proofs}. 
    Although (a2) allows one to leverage robustness properties of contracting systems~\cite[Lem. 1]{DDV-JJES:13}, (a2) may not be satisfied for several systems, including LTI systems.
    Next, Theorem~3 of~\cite{DDV-JJES:13} requires contractivity of the original fast and slow sub-systems, which is stronger than our assumptions that the reduced model and fast sub-system are contractive. Our assumptions more closely mirrors the Lyapunov-based approach; see Remark~\ref{rem:comp_Lyap} for a detailed discussion on this. 
    Moreover, contractivity of both the fast and slow sub-systems are in general not satisfied in the OFO example motivated earlier~\ref{sec:extensions-and-applications}. %examples we present in Section~\ref{sec:extensions-and-applications}.
  %  Finally, our Theorem~\ref{thm:str-contractive-whole-system-general} provides explicit conditions for the contractivity of the shifted dynamics based on system properties such as Lipschitz constants and bounds on the disturbances. In contrast, the conditions in~\cite{DDV-JJES:11}
   % are in terms of the largest singular value of the generalized Jacobian matrix of the two-time scale system. 
    \hfill $\triangle$
\end{remark}

% \vspace{.1cm}

\subsection{Autonomous Two-Time Scale Dynamical Systems}
\label{sec:simplified-nonlinear-dynamics}
% \textcolor{red}{In this section, ... [work out general case from before...]}
In this section, we specialize the results of the general setting in Section~\ref{sec:general-nonlinear-dynamics-with-disturbances}.
We consider the \textit{two-time scale autonomous system},
\begin{subequations}\label{eq:aut-gen}
    \begin{align}\label{eq:aut-gen-f}
        \dot x &= f(x, z), \hspace{1.0cm} x(0) = x_0\\
        \eps \dot z &= g(x, z) , \hspace{1.0cm} z(0) = z_0 
        \label{eq:aut-gen-g}
    \end{align}
\end{subequations}
where the vector fields $f: \R^{n_x} \times \R^{n_z} \to \R^{n_x}$, $g: \R^{n_x} \times \R^{n_z} \to \R^{n_z}$ are continuous in their arguments. % and $f$ is locally Lipschitz in $x \in \R^{n_x}$. 
Here, $\eps$ is a parameter that enforces a time-scale separation, but is not an argument of the vector fields.    
As before, assume there exists (possibly different) norms $\|\cdot\|_{x}$ and $\|\cdot\|_{z}$ with compatible weak pairings $\llbracket \cdot ; \cdot \rrbracket_{x}$ and $\llbracket \cdot ; \cdot \rrbracket_{z}$ 
on $\R^{n_x}$ and $\R^{n_z}$, respectively.
In the following, we simplify 
the results in Section~\ref{sec:general-nonlinear-dynamics-with-disturbances} 
to analyze the two-time scale autonomous system~\eqref{eq:aut-gen}. 
We start by customizing Assumptions~\ref{as:fast-str-contracting}-\ref{as:Lipschitz-interconnection-x-z} 
% and~\ref{as:extra-regularity}
used in Section~\ref{sec:general-nonlinear-dynamics-with-disturbances}.

\vspace{.1cm}

\begin{assumption}[Contractivity of the fast dynamics]
\label{as:aut-fast-str-contracting}
~

\noindent
    There exists $c_g > 0$ such that $\osL_z(g(x,z)) \leq -c_g$ holds uniformly in $x.$     \hfill $\triangle$
\end{assumption}

\vspace{.1cm}

\begin{assumption}[Continuity of the slow dynamics]
\label{as:aut-Lipschitz-f-in-t}
~

\noindent
There exists $\ell_{f_x} \geq 0$ such that $\Lip_x(f(x,z)) \leq \ell_{f_x}$ holds uniformly in $z$.  
\hfill $\triangle$ 
\end{assumption}

\vspace{.1cm}

\begin{assumption}[Lipschitz interconnection]
\label{as:aut-Lipschitz-interconnection-x-z}
~

\noindent
There exists $\ell_{f_z} > 0$ such that $\Lip_z(f(x,z)) \leq \ell_{f_z}$ holds uniformly in $x$.  There exists $\ell_{g_x} > 0$ such that $\Lip_x(g(x,z)) \leq \ell_{g_x}$ holds uniformly in $z.$ \hfill $\triangle$
\end{assumption}

%%%%%%%% REMOVING BC WE DON'T NEED ANYMORE: 
% \vspace{.1cm}

% \begin{assumption}[Differentiability of $z^*(x)$]
% \label{as:extra-regularity_auto} 
% ~

% \noindent
% For given $x \in \R^{n_x}$ and $\eps \geq 0$, let $z^*(x)$ be the unique solution to the  equation $g(x,z^*(x)) = 0$. The map  $x \mapsto z^*(x)$ is globally differentiable. \hfill $\triangle$
% \end{assumption}

\vspace{.1cm} 

We reiterate that
Assumptions~\ref{as:aut-fast-str-contracting}-\ref{as:aut-Lipschitz-interconnection-x-z}
ensure that solutions of the two-time scale autonomous system~\eqref{eq:aut-gen} exist and are unique, with
Assumption~\ref{as:aut-fast-str-contracting} further implying that, for any
$\eps > 0$,~\eqref{eq:aut-gen-g} is strongly infinitesimally contractive
with rate $c_g / \eps$ in the domain $t$. 
We rewrite the two-time scale autonomous system~\eqref{eq:aut-gen} in the new variables to define the \textit{shifted autonomous dynamics}, given by: 
\begin{subequations}
\label{eq:aut-new-var-w}
    \begin{align}\label{eq:aut-new-var-w-a}
        \dot x &= f(x, y + z^*(x)) \\
            \eps \dot y &= g(x, y + z^*(x)) - \eps \frac{\partial z^*}{\partial x}(x) f(x, y + z^*(x)) \label{eq:aut-new-var-w-b}
    \end{align}
\end{subequations}
with initial conditions $x(0) = x_0$ and $y(0) = z_0 - z^*(x_0)$, respectively. 
The \textit{reduced autonomous model} is given by:
    \begin{align}\label{eq:aut-reduced}
        \dot x_{\text{r}} &= f(x_{\text{r}}, z^*(x_{\text{r}})), ~~~~~~~ x_{\text{r}}(0) = x(0) = x_{0}
    \end{align}
where the fast variable $z$ is substituted by the quasi-steady state $z^*(x_{\text{r}})$. 
We impose the following assumption on the reduced autonomous model~\eqref{eq:aut-reduced}.

\vspace{.1cm}

\begin{assumption}[Contractivity of the reduced aut. model]
\label{as:red-str-contracting-aut}
%     The reduced autonomous model~\eqref{eq:aut-reduced} is strongly infinitesimally contracting with rate $c_f > 0$; that is, $\osL_x(f(x, z^*(x))) \leq -c_f$. \hfill $\triangle$ 
There exists $c_f > 0$ such that $\osL_x(f(z^*(x))) \leq -c_f$. \hfill $\triangle$ 
\end{assumption}

\vspace{.1cm}

The following is the simplified counterpart to Theorem~\ref{thm:closeness-of-solns-general}.   

\begin{corollary}[Closeness of the solutions]
\label{thm:closeness-of-solns-aut}
Consider the two-time scale autonomous system~\eqref{eq:aut-gen} satisfying (contractivity) Assumptions~\ref{as:aut-fast-str-contracting} and~\ref{as:red-str-contracting-aut}, and (continuity) Assumptions~\ref{as:aut-Lipschitz-f-in-t}-\ref{as:aut-Lipschitz-interconnection-x-z}. 
%%%%%%%% REMOVING BC WE DON'T NEED ANYMORE: 
%, 
% and (differentiability of $z^*$) Assumption~\ref{as:extra-regularity_auto}.  
Suppose that, 
\begin{align}
\label{eq:a0esstart}
0 < \eps < \eps^*_{0,\text{as}} =  \frac{c_g^2}{\ell_{g_x} \ell_{f_z}}.
\end{align}
Then, for any initial condition, the two-time scale autonomous system~\eqref{eq:aut-gen} has a unique solution $x(t), z(t)$, for $t \geq 0$, that satisfies the bounds
\allowdisplaybreaks
\begin{subequations}\label{eq:thm-close-solns-aut}
\begin{align}
        \|x(t) - \subscr{x}{r}(t)\|_x &\leq  G_x(t) \label{eq:thm-close-solns-aut-a}, \\
        \|z(t)  - z^*(\subscr{x}{r}(t))\|_z &\leq G_z(t),
     \label{eq:thm-close-solns-aut-b} 
\end{align}
\end{subequations}
where $t \mapsto \subscr{x}{r}(t)$ is a solution of the reduced autonomous model~\eqref{eq:aut-reduced} 
and the maps $t \mapsto G_x(t)$ and $t \mapsto G_z(t)$ 
satisfy $G_x(t) \rightarrow 0$, $G_z(t) \rightarrow 0$ for $t \rightarrow +\infty$. 
\QEDB
\end{corollary}

\vspace{.1cm}

The bounds~\eqref{eq:thm-close-solns-aut} demonstrate asymptotic convergence of the trajectories of the two-time scale autonomous dynamics~\eqref{eq:aut-gen} to the trajectories of its limiting system; that is, $\|x(t) - \subscr{x}{r}(t)\|_x \rightarrow 0$ and  $\|z(t) - z^*(\subscr{x}{r}(t))\|_z \rightarrow 0$ as $t \rightarrow + \infty$. This is in contrast to Theorem~\ref{thm:closeness-of-solns-general},  where the trajectory of the original dynamics converge within a ball centered on $\subscr{x}{r}(t)$ and $z^*(\subscr{x}{r}(t))$. 
In the following, we provide the expressions for $G_x(t)$ and $G_z(t)$, which are derived for two different cases. To this end, we first define the positive constant $c_y := c_g \eps^{-1} - \frac{\ell_{g_x}}{c_g} \ell_{f_z}$.

\noindent \emph{Case~1}: $c_f \neq c_y$. 
\begin{align*}
    &G_x(t) : = \ell_{f_z} \|z(0) - z^*(x(0))\|_z \frac{e^{-c_y t} - e^{-c_f t}}{c_f - c_y} \\
    & \quad\quad\quad\quad\quad\quad+ \ell_{f_z} \frac{\ell_{g_x}}{c_g} \| f(x(0), z^*(x(0))\|_x \Big( \frac{t e^{-c_f t}}{c_y - c_f} + \frac{ e^{-c_f t} - e^{-c_y t} }{(c_f - c_y)^2} \Big), \\
    &G_z(t) : = e^{-c_y t} \|z(0) - z^*(x(0))\|_z  +\frac{\ell_{g_x}}{c_g} G_x(t) \\
    & \quad\quad\quad\quad\quad\quad+ \frac{\ell_{g_x}}{c_g} \|f(x(0), z^*(x(0)))\|_x \Big( \frac{e^{-c_f t} - e^{-c_y t}}{c_y - c_f} \Big).
\end{align*}

\noindent \emph{Case~2}: $c_f = c_y$.  
\begin{align*}
    &G_x(t) := \ell_{f_z} \big(\|z(0) - z^*(x(0))\|_z t e^{-c_f t} + \frac{1}{2} t^2 e^{-c_f t}\big)\\
    &G_z(t) := e^{-c_y t} \|z(0) - z^*(x(0))\|_z  + \frac{\ell_{g_x}}{c_g}\|f(x(0), z^*(x(0))\|_x t e^{-c_y t} 
    + \frac{\ell_{g_x}}{c_g} G_x(t).
\end{align*}

The proof for Corollary~\ref{thm:closeness-of-solns-aut}  closely follows that of Theorem~\ref{thm:closeness-of-solns-general} given in Section~\ref{sec:proofs} and is omitted due to space limitations.

\section{Proofs}\label{sec:proofs}

We first introduce some additional notation and preliminary results used throughout this section. 
Throughout our proofs, let the norms $\| \cdot \|_x$ and $\| \cdot \|_z$ on $\R^{n_x}$ and $\R^{n_z}$ have compatible weak pairings $\llbracket \cdot ; \cdot \rrbracket_{x}$ and $ \llbracket \cdot ; \cdot \rrbracket_{z}$.

We let $\|\cdot\|$ be a norm with compatible weak pairing $\llbracket \cdot ; \cdot \rrbracket: \R^n \times \R^n \rightarrow \R$ on $\R^n$~\cite[Ch.~2]{FB:24-CTDS}. 
We recall the following two properties of weak pairings on normed spaces, which will be used in the proofs: 

\noindent \emph{(P1) The curve norm derivative formula:} 
for every differentiable curve $x: (a,b) \rightarrow \R^n$ and for almost every $t \in (a,b)$, $\|x(t)\| D^+ \|x(t)\| = \llbracket \dot x(t) ; x(t) \rrbracket$.

\noindent \emph{(P2) Sub-additivity and continuity of first argument:} $\llbracket x_1 + x_2 ; y \rrbracket \leq \llbracket x_1 ; y \rrbracket + \llbracket x_2 ; y \rrbracket$ for all $x_1, x_2, y \in \R^n$, where $\llbracket \cdot ; \cdot \rrbracket$ is continuous in its first argument. 

\noindent \emph{(P3) Cauchy-Swartz Inequality:}
$\lvert \llbracket x; y \rrbracket \lvert \leq \llbracket x ; x \rrbracket^{1/2} \llbracket y ; y \rrbracket^{1/2} $ for all $x,y \in \R^n$.
% Our proofs will leverage the following lemma about the unique solution to $g(x,z,w_z,\eps) = 0$. 
% \vspace{.1cm} 
% \begin{lemma}[Lipschitz solutions to $g(x,z,w_z,\eps) = 0$]
% \label{lem:lipz}
% ~
% \noindent
% Let Assumptions~\ref{as:fast-str-contracting} and~\ref{as:Lipschitz-interconnection-else} hold, and let $z^*(x, w_z,\eps)$ be the unique solution to 
%     $g(x,z,w_z,\eps) = 0$ for given $x, w_z$ and $\eps$. Then, the following holds: 
% \noindent \emph{i)} for any fixed $w_z$ and $\eps$, the map $x \mapsto z^*(x,w_z,\eps)$ is Lipschitz with constant $\frac{\ell_{g_x}}{c_g}$. 
% \noindent \emph{ii)} For any fixed $x$ and $\eps$, the map $w_z \mapsto z^*(x, w_z,\eps)$ is Lipschitz with constant $\frac{\ell_{g_w}}{c_g}$. 
% \noindent \emph{iii)} For any fixed $x$ and $w_z$, the map $\eps \mapsto z^*(x, w_z,\eps)$ is Lipschitz with constant $\frac{\ell_{g_\eps}}{c_g}$. 
% \QEDB
% \end{lemma}
% \vspace{.1cm} 
% The Lipschitz constants in Lemma~\ref{lem:lipz} can be derived using the arguments of~\cite[Lem. 1]{AD-VC-AG-GR-FB:23f}. The next lemma focuses on the shifted dynamics~\eqref{eq:gen-new-var-w-b}. 
% \vspace{.1cm} 
%\textcolor{red}{fix questions}

%\noindent
Lemma~\ref{lem:bound-y} focuses on the shifted dynamics~\eqref{eq:gen-new-var-w-b}.
\begin{lemma}[Transient bound for $\|y(t)\|_z$]\label{lem:bound-y}
Consider the fast sub-system~\eqref{eq:gen-new-var-w-b} in the shifted variable $y = z - z^*(x,w_z)$, 
satisfying Assumptions~\ref{as:fast-str-contracting} and~\ref{as:Lipschitz-interconnection-x-z}. Let $\eps$ be such that $0<\eps < c_g^2 / (\ell_{g_x} \ell_{f_z})$. Then, any unique solution $t \mapsto y(t)$ of~\eqref{eq:gen-new-var-w-b} satisfies, 
\begin{align}
    \|y(t)\|_z \leq e^{-c_y t} \|y(0)\|_z +  \frac{\delta}{c_y} + E_y(t, \eps), ~~~~ t \geq 0,
\end{align}
with $\delta_y := \frac{\ell_{g_x}}{c_g c_f} \big( \ell_{f_t} + \ell_{f_z} \frac{\ell_{g_w}}{c_g} \bar w_z + \ell_{f_w} \bar w_x \big)$,  $\delta := \ell_{g_\eps} + \frac{\ell_{g_w}}{c_g} \bar w_z + \eps \frac{\ell_{g_x}}{c_g} \ell_{f_\eps}  + \delta_y$, $c_y := c_g \eps^{-1} - \frac{\ell_{g_x}}{c_g} \ell_{f_z}$, and where  
the map $E_y(t, \eps) \to 0$ as $t \to \infty$. Moreover,  we have that:
\begin{align}\label{eq:lim-sup-y}
    \limsup_{t \to \infty} \|y(t)\|_z \leq \frac{\delta}{c_y}. 
\end{align} 
\hfill $\square$
\end{lemma}

\begin{proof}
We will compute a bound for 
$\|y(t)\|_z \eps D^+ \|y(t)\|_z = \|z(t) - z^*(x(t), w_z(t))\|_z$ $ \eps D^+ \|z(t) - z^*(x(t), w_z(t))\|_z,$ and then apply the Gr\"{o}nwall Inequality. First, use property (P1) of the weak pairing to compute 
    \begin{align*}
        \|z(t) - z^*(x(t), w_z(t))\|_z \eps D^+ &\|z(t) - z^*(x(t), w_z(t))\|_z \\
        & \quad \overset{\text{(P1)}}{=} \llbracket g(x, z, w_z, \eps) - \eps \frac{d}{dt} z^*(x, w_z); z - z^*(x, w_z)) \rrbracket_z.
    \end{align*}
    Add and subtract $g(x, z,w_z, 0)$,  subtract $g(x, z^*(x, w_z), w_z, 0)=0$ (for any given $w_z$), and then apply the properties (P2)-(P3) to obtain: 
    \begingroup
    \allowdisplaybreaks
    \begin{align*}
        \|z(t) - &z^*(x(t), w_z(t))\|_z \eps D^+  \|z(t) - z^*(x(t), w_z(t))\|_z
      \\
      &\quad  \leq \llbracket g(x, z,w_z, 0) - g(x, z^*(x, w_z),w_z, 0); z - z^*(x, w_z)\rrbracket_z  \\
      &  \quad\quad + \llbracket g(x, z, w_z, \eps) - g(x, z,w_z, 0) ; z - z^*(x, w_z) \rrbracket_z \\
      & \quad\quad + \eps \left\|\frac{d}{dt} z^*(x, w_z) \right\|_z \|z - z^*(x,w_z)\|_z \\
      & \quad\leq -c_g \|z - z^*(x, w_z)\|_z^2
      + \eps \Big( \ell_{g_\eps} + \left\|\frac{d}{dt} z^*(x, w_z) \right\|_z \Big)\|z - z^*(x, w_z)\|_z. 
    \end{align*}
    \endgroup
    Next, bound the norm of $\frac{d}{dt} z^*(x, w_z)$ as, 
    \begin{align*}
      \left\|\frac{d}{dt} z^*(x, w_z) \right\|_z &= \left\|\frac{\partial
        z^*}{\partial w_z} \dot w_z + \frac{\partial z^*}{\partial x} f(t,
      x, z, w_x, \eps) \right\|_z \\
      &\leq \frac{\ell_{g_w}}{c_g} \|\dot
      w_z\|_{w_z} + \frac{\ell_{g_x}}{c_g} \|f(t, x, z, w_x, \eps) \|_{x}.
    \end{align*}
    Consider
    \begin{align*}
        \|f(t, x, z, w_x, \eps) \|_x &\leq \|f(t, x, z, w_x, \eps) - f(t, x, z, w_x, 0)\|_x \\
        & \quad +  \|f(t, x, z, w_x, 0) - f(t, x, z^*(x, w_z), w_x, 0)\|_x \\
        &  \quad + \|f(t, x, z^*(x, w_z), w_x, 0)\|_x\\
        & \leq \ell_{f_\eps} \eps + \ell_{f_z} \|z - z^*(x, w_z)\|_z + \|f(t, x, z^*(x, w_z), w_x, 0)\|_x.
    \end{align*}
    Then, we use the proof of~\cite[Thm. 3.9]{FB:24-CTDS} to bound $\|f(t, x, z^*(x, w_z), w_x, 0)\|_x$.
    Before doing so, note that $\subscr{f}{red}(t, x, w_x, w_z) := f(t, x, z^*(x, w_z), w_x, 0)$ is Lipschitz in $w_z$ with constant $\ell_{f_z} \frac{\ell_{g_w}}{c_g}$. Compute
    \begin{align}
       \|\subscr{f}{red}&(t, x, w_x, w_z) \|_x D^+ \|\subscr{f}{red}(t, x, w_x, w_z)\|_x  \nonumber \\
      &\overset{\text{by (P1)}}{=} \left\llbracket\frac{d}{dt} \subscr{f}{red}(t, x, w_x, w_z) ; \subscr{f}{red}(t, x, w_x, w_z) \right\rrbracket_x \nonumber \\
       &~~~= \left\llbracket \frac{\partial \subscr{f}{red}}{\partial t} + \frac{\partial \subscr{f}{red}}{\partial x} \subscr{f}{red} + \frac{\partial \subscr{f}{red}}{\partial z}  \frac{\partial z^*}{\partial w_z}\dot w_z + \frac{\partial \subscr{f}{red}}{\partial w_x}  \dot w_x ; \subscr{f}{red}(t, x, w_x, w_z) \right\rrbracket_x \nonumber \\
       &~~~\leq -c_f \|\subscr{f}{red}(t, x, w_x, w_z)\|_x^2 + \Big( \ell_{f_t} + \ell_{f_z} \frac{\ell_{g_w}}{c_g} \bar w_z + \ell_{f_w} \bar w_x \Big) \|\subscr{f}{red}(t, x, w_x, w_z)\|_x, \nonumber
    \end{align}
    where $\bar w_z := \operatorname{ess~sup}_{t \geq 0} \|\dot w_z(t)\|_{w_z}$ (and similarly for $\bar w_x$).
    Apply the Gr\"{o}nwall Inequality to obtain
    \begin{align}
        \|\subscr{f}{red}(t, x, w_x, w_z)\|_x &\leq e^{-c_f t} \|\subscr{f}{red}(0, x, w_x, w_z)\|_x  \\
        & \quad\quad + \Big( \ell_{f_t}  + \ell_{f_z} \frac{\ell_{g_w}}{c_g} \bar w_z + \ell_{f_w} \bar w_x \Big) \frac{1}{c_f} (1 - e^{-c_f t}). \nonumber
    \end{align}   
    Then, the time derivative of the steady state is bounded as
    \begin{align*}
        \Big\|&\frac{d}{dt} z^*(x, w_z) \Big\|_x \leq \frac{\ell_{g_w}}{c_g} \|\dot w_z\|_{w_z} + \frac{\ell_{g_x}}{c_g} \ell_{f_\eps} \eps  
         + \frac{\ell_{g_x}}{c_g} \ell_{f_z} \|z - z^*(x, w_z)\|_z 
         \\ & \quad + \frac{\ell_{g_x}}{c_g} e^{-c_f t} \|\subscr{f}{red}(0, x, w_x, w_z)\|_x + \frac{\ell_{g_x}}{c_g c_f} \Big( \ell_{f_t} + \ell_{f_z} \frac{\ell_{g_w}}{c_g} \bar w_z + \ell_{f_w} \bar w_x \Big)  (1 - e^{-c_f t}).
    \end{align*}
    Accordingly, the bound on $\eps D^+\|z - z^*(x, w_z)\|_z$ is
    \begin{multline*}
        \eps D^+\|z - z^*(x, w_z)\|_z \leq -c_g \|z - z^*(x, w_z)\|_z + \eps  \ell_{g_\eps} 
        + \eps \frac{\ell_{g_x}}{c_g}e^{-c_f t} \|\subscr{f}{red}(0, x, w_x, w_z)\|_x 
        \\
         \null\quad + \eps \Big( \frac{\ell_{g_w}}{c_g} \bar w_z + \frac{\ell_{g_x}}{c_g} \big(\ell_{f_\eps} \eps + \ell_{f_z} \|z - z^*(x, w_z)\|_z \big) \Big)   \\         
         + \eps \frac{\ell_{g_x}}{c_g c_f} \Big( \ell_{f_t}  + \ell_{f_z} \frac{\ell_{g_w}}{c_g} \bar w_z + \ell_{f_w} \bar w_x \Big)  (1 - e^{-c_f t}). 
    \end{multline*}
    For brevity, define $c_y := c_g \eps^{-1} - \frac{\ell_{g_x}}{c_g} \ell_{f_z}$ and $\delta_y := \frac{\ell_{g_x}}{c_g c_f} \big( \ell_{f_t} + \ell_{f_z} \frac{\ell_{g_w}}{c_g} \bar w_z + \ell_{f_w} \bar w_x \big).$
    We now consider two cases. 

    \noindent 
    \textit{Case 1: $c_f \neq c_y$.} 
    Apply the Gr\"{o}nwall Bellman Inequality to obtain     
    \begin{align}
        \|y(t)&\|_z \leq e^{-c_y t} \|y(0)\|_z  + \frac{\ell_{g_\eps} + \frac{\ell_{g_w}}{c_g} \bar w_z + \eps \frac{\ell_{g_x}}{c_g} \ell_{f_\eps} + \delta_y }{c_y} \nonumber \\
        & - e^{-c_y t} \frac{\ell_{g_\eps} + \frac{\ell_{g_w}}{c_g} \bar w_z + \eps \frac{\ell_{g_x}}{c_g} \ell_{f_\eps} + \delta_y}{c_y} + \Big(\frac{\ell_{g_x}}{c_g} \|\subscr{f}{red}(0)\|_x - \delta_y \Big) \Big(\frac{e^{-c_f t} - e^{-c_y t}}{c_y - c_f} \Big)
    \end{align}
    where $\|\subscr{f}{red}(0)\|_x$ is a short-hand notation for $\|\subscr{f}{red}(0, x, w_x, w_z)\|_x$. 
    Here, $\delta := \ell_{g_\eps} + \frac{\ell_{g_w}}{c_g} \bar w_z + \eps \frac{\ell_{g_x}}{c_g} \ell_{f_\eps} + \delta_y$ and $E_y(t, \eps) := - e^{-c_y t} \frac{\ell_{g_\eps} + \frac{\ell_{g_w}}{c_g} \bar w_z + \eps \frac{\ell_{g_x}}{c_g} \ell_{f_\eps} + \delta_y}{c_y} + \Big(\frac{\ell_{g_x}}{c_g} \|\subscr{f}{red}(0)\|_x - \delta_y \Big) \Big(\frac{e^{-c_f t} - e^{-c_y t}}{c_y - c_f} \Big).$

    \noindent 
    \textit{Case 2: $c_f = c_y$.} 
   In this case we obtain, 
    \begin{align}
        \|y(t)\|_z &\leq e^{-c_y t} \|y(0)\|_z  + \frac{\ell_{g_\eps} + \frac{\ell_{g_w}}{c_g} \bar w_z + \eps \frac{\ell_{g_x}}{c_g} \ell_{f_\eps} + \delta_y }{c_y} \nonumber \\
        & \quad \quad - e^{-c_y t} \frac{\ell_{g_\eps} + \frac{\ell_{g_w}}{c_g} \bar w_z + \eps \frac{\ell_{g_x}}{c_g} \ell_{f_\eps} + \delta_y}{c_y} + \Big(\frac{\ell_{g_x}}{c_g} \|\subscr{f}{red}(0)\|_x - \delta_y \Big) \big(t e^{-c_y t} \big).
    \end{align}
    Here, $\delta := \ell_{g_\eps} + \frac{\ell_{g_w}}{c_g} \bar w_z + \eps \frac{\ell_{g_x}}{c_g} \ell_{f_\eps} + \delta_y$ and 
    $E_y(t, \eps) := - e^{-c_y t} \frac{\ell_{g_\eps} + \frac{\ell_{g_w}}{c_g} \bar w_z + \eps \frac{\ell_{g_x}}{c_g} \ell_{f_\eps} + \delta_y}{c_y} + \Big(\frac{\ell_{g_x}}{c_g} \|\subscr{f}{red}(0)\|_x - \delta_y \Big) \big(t e^{-c_y t} \big).$
    The bound~\eqref{eq:lim-sup-y} follows from the above computations.
\end{proof}

\subsection{Proof of Theorem~\ref{thm:closeness-of-solns-general}}

\noindent
\emph{(a) Bound on $\|x(t) - x_\text{r}(t)\|$. } We first prove the bound~\eqref{eq:thm-close-solns-general-a}. 
Consider the slow dynamics~\eqref{eq:gen-system-a}, and substitute $z = y + z^*(x, w_z)$ for its argument. Then, rewrite the dynamics~\eqref{eq:gen-system-a} as,
\begin{align}
    \dot x = &f(t,x, z^*(x,w_z), w_x, 0) + \underbrace{f(t,x, y + z^*(x,w_z), w_x, \eps) - f(t,x, z^*(x,w_z), w_x, \eps)}_{:= d(x,y)} \nonumber \\
    &\quad + \underbrace{f(t,x, z^*(x,w_z),w_x, \eps) - f(t,x, z^*(x,w_z), w_x, 0)}_{:=d_\eps(x,y_\text{bl})}, \label{eq:rewritten-slow-gen-dyn} % \label{eq:gen-a-reshifted-simple}
\end{align}
% \sloppy makes it so that the in line math doesn't spill over the edges
\sloppy which are obtained by adding and subtracting $f(t,x, z^*(x,w_z), w_x, 0)$ and $f(t,x, z^*(x,w_z), w_x, \eps)$. 
\sloppy From Assumption \ref{as:Lipschitz-interconnection-x-z} we have that $y \mapsto f(t,x, y + z^*(x,w_z), w_x,\eps)$ is Lipschitz continuous with constant $\ell_{f_z} \geq 0$; in fact, for fixed $\eps$, $x$, $t$, $w_x$, $w_z$, and for any $y_1, y_2 \in \R^{n_z}$, one may verify that
\begin{align*}
    \|f(t, x, y_1 + z^*(x,w_z),& w_x,\eps)  - f(t, x, y_2 + z^*(x,w_z), w_x,\eps)\|_x  \\
    &\leq \ell_{f_z} \|y_1 - z^*(x,w_z) - y_2 + z^*(x,w_z)\|_z = \ell_{f_z} \|y_1 - y_2\|_z.
\end{align*} 

Let $x(t)$ and $\subscr{x}{r}(t)$ be solutions of the slow dynamics \eqref{eq:rewritten-slow-gen-dyn} and of the reduced model \eqref{eq:gen-reduced}, respectively. We bound $\|x(t) - \subscr{x}{r}(t)\|_x D^{+} \|x(t) - \subscr{x}{r}(t)\|_x$ as follows: 
\begingroup
\allowdisplaybreaks
\begin{align*}
    &\|x(t) - \subscr{x}{r}(t)\|_x D^{+} \|x(t) - \subscr{x}{r}(t)\|_x \\
    &\overset{\text{by (P1)}}{=} \llbracket f(t, x,z^*(x,w_z),w_x,0) + d(x, y) + d_\eps(x)  - f(t, \subscr{x}{r},z^*(\subscr{x}{r},w_z),w_x,0); x - \subscr{x}{r}\rrbracket_x  \\
    &\overset{\text{by (P2)}}{\leq} \llbracket f(t, x,z^*(x,w_z),w_x,0)  - f(t, \subscr{x}{r},z^*(\subscr{x}{r},w_z),w_x,0); x - \subscr{x}{r}\rrbracket_x \\
    & \quad\quad\quad\quad  + \llbracket d(x, y); x - \subscr{x}{r}\rrbracket_x + \llbracket d_\eps(x, y); x - \subscr{x}{r}\rrbracket_x\\
    &\quad\leq -c_f \|x - \subscr{x}{r}\|_x^2 + \eps\ell_{f_\eps}\|x - \subscr{x}{r}\|_x + \ell_{f_z} \Big(e^{-c_y t}\|y(0)\|_z + \frac{\delta}{c_y} + E_y(t, \eps) \Big) \|x - \subscr{x}{r}\|_x,
\end{align*}
\endgroup
where we used the properties (P1)-(P2) of the weak pairings, apply Assumption \ref{as:red-str-contracting} to the last inequality, and use Lemma~\ref{lem:bound-y}, recalling that the constants $\delta_y, \delta$, and $c_y$ and the map $E_y(t,\eps)$ are identified in Lemma~\ref{lem:bound-y}.
Moreover, we used the computations 
\begin{align*}
    \llbracket d(x, y); x - &\subscr{x}{r}\rrbracket_x = \llbracket f(t,x, y + z^*(x,w_z),w_x,\eps) - f(t,x,z^*(x,w_z),w_x,\eps); x - \subscr{x}{r}\rrbracket_x \\
    &\leq \ell_{f_z} \|y\|_z \|x - \subscr{x}{r}\|_x \overset{\text{Lem.~\ref{lem:bound-y}}}{\leq} \ell_{f_z} \Big(e^{-c_y t} \|y(0)\|_z +  \frac{\delta}{c_y} + E_y(t, \eps) \Big) \|x - \subscr{x}{r}\|_x,
\end{align*}
as well as
\begin{align*}
    \llbracket d_\eps(x, y); x - \subscr{x}{r}\rrbracket_x &= \llbracket f(t, x,z^*(x,w_z), w_x,\eps)  - f(t, x,z^*(x,w_z),w_x,0); x - \subscr{x}{r}\rrbracket_x \\
    &\leq \ell_{f_\eps} \eps \|x - \subscr{x}{r}\|_x,
\end{align*}
where $\eps \geq 0$. 
Thus, dividing by $\|x(t) - \subscr{x}{r}(t)\|_x \neq 0$, we obtain the bound
\begin{align*}
    D^{+} \|x(t) - \subscr{x}{r}(t)\|_x \leq& -c_f \|x(t) - \subscr{x}{r}(t)\|_x + \eps \ell_{f_\eps}+\ell_{f_z} \Big(e^{-c_y t} \|y(0)\|_y +  \frac{\delta}{c_y} + E_y(t, \eps) \Big)\, .
\end{align*}
We next consider two different cases. 

\emph{Case~1}: $c_f \neq c_y$. 
Applying the Gr\"{o}nwall Comparison Lemma  and~\cite[Cor. 3.17]{FB:24-CTDS} we have:
\begin{align}\label{eq:x-xr}
     \|x(t) - \subscr{x}{r}(t)\|_x \leq& e^{-c_f t} \|x(0) - \subscr{x}{r}(0)\|_x \nonumber  \\
     & + \Big(\eps \ell_{f_\eps} + \ell_{f_z} \frac{\delta}{c_y} \Big)(1 - e^{-c_f t})  +  \big(\ell_{f_z} \|y(0)\|_z - \frac{ \delta}{c_y} \big)\Big(\frac{e^{-c_y t} - e^{-c_f t}}{c_f - c_y} \Big)\\
    & + \Big( \frac{\ell_{g_x}}{c_g}\|\subscr{f}{red}(0)\|_x - \delta_y\Big) \Big( \frac{e^{-c_f t} - e^{-c_y t}}{(c_y - c_f)^2}-  \frac{t e^{-c_y t}}{(c_y - c_f)} \Big) \, .\nonumber
\end{align}
The bound in~\eqref{eq:thm-close-solns-general-a} is then obtained by defining the function $E_x(t,\eps)$ as:
\begin{align}\label{eq:E-x-t-eps-function}
    E_x (t,\eps) := & \Big(\ell_{f_z} \|y(0)\|_z - \frac{ \delta}{c_y} \Big) \Big(\frac{e^{-c_y t} - e^{-c_f t}}{c_f - c_y} \Big)  - \Big(\ell_{f_\eps} \ell_{f_z}\frac{\delta}{c_g} \Big) e^{-c_f t} \nonumber \\
    & + \Big( \frac{\ell_{g_x}}{c_g}\|\subscr{f}{red}(0)\|_x - \delta_y\Big) \Big( \frac{e^{-c_f t} - e^{-c_y t}}{(c_y - c_f)^2}-  \frac{t e^{-c_y t}}{(c_y - c_f)} \Big) \, .\nonumber
\end{align}

\emph{Case~2}: $c_f = c_y$.  Applying the Gr\"{o}nwall Comparison Lemma and re-arranging terms,  $E_x(t,\eps)$ is given by 
\begin{align}\label{eq:E-x-t-eps-function-2}
E_x (t,\eps) &:= \big(\ell_{f_z} \|y(0)\|_z - \frac{ \delta}{c_y} \big) t e^{-c_f t}  - \big( \ell_{f_z} \frac{\delta}{c_g} + \ell_{f_\eps} \big)e^{-c_f t} \\
& \quad\quad + \Big( \frac{\ell_{g_x}}{c_g}\|\subscr{f}{red}(0)\|_x - \delta_y\Big) \frac{1}{2} t^2 e^{-c_y t} \, . \nonumber
\end{align}

\noindent
\textit{(b) Bound on} $\|z(t) - z^*(\subscr{x}{r}(t), w_z(t))\|_z$. We recall that $z(t)$ is a solution of the original system \eqref{eq:gen-system-b}, and $\subscr{x}{r}(t)$ is a solution to the reduced model~\eqref{eq:gen-reduced}.

We first consider this inequality:
\begin{align}\label{eq:fast-var-tri-split}
    \|z  - z^*(\subscr{x}{r}, w_z)\|_z \leq & \|z  - z^*(x, w_z)\|_z + \|z^*(x, w_z) - z^*(\subscr{x}{r}, w_z)\|_z,
\end{align}
and analyze each term separately. Since $y := z - z^*(x, w_z)$, the first bound is just Lemma~\ref{lem:bound-y}. For the second bound of \eqref{eq:fast-var-tri-split}, we use Lemma~\ref{lem:lipz} to obtain $\|z^*(x, w_z) - z^*(\subscr{x}{r}, w_z)\|_z \leq \frac{\ell_{g_x}}{c_g} \|x - \subscr{x}{r}\|_x$ and then bound $\|x - \subscr{x}{r}\|_x$ with~\eqref{eq:thm-close-solns-general-a}.
In total the bound on $\|z - z^*(\subscr{x}{r}, w_z)\|_z$ is given as, 
\begin{align}
\|z - z^*(\subscr{x}{r}, w_z)\|_z \leq &
     \frac{\delta}{c_y}  +\frac{\ell_{g_x}}{c_g} \Big(\eps \ell_{f_\eps} + \frac{\ell_{f_z} \delta}{c_y}\Big) + e^{-c_y t} \|y(0)\|_z \nonumber \\
     & + E_y(t, \eps)  + \frac{\ell_{g_x}}{c_g} \Big( e^{-c_f t} \|x(0) - \subscr{x}{r}(0)\|_x + E_x(t,\eps) \Big) ,
\end{align}
where setting $E_z(t, \eps)$ as,  
\begin{align}
    E_z(t, \eps) := e^{-c_y t} \|y(0)\|_z + E_y(t, \eps) + \frac{\ell_{g_x}}{c_g} \Big( e^{-c_f t} \|x(0) - \subscr{x}{r}(0)\|_x + E_x(t,\eps) \Big) 
\end{align}
yields the result.
This concludes the proof of the theorem.

\section{Results for OFO}\label{sec:results-for-OFO}

In this section, we apply our results of Section~\ref{sec:main-results} to the motivating example~\ref{sec:extensions-and-applications}. 
We continue our analysis of the closed-loop system~\eqref{eq:ex-LTI-grad} by 
following the structure of Section~\ref{sec:main-results}.
As such, the shifted dynamics are given by: 
\begin{subequations}\label{eq:ex-LTI-grad-shifted}
    \begin{align}\label{eq:ex-LTI-grad-shifted-b}
        \dot u &= - \nabla \phi(u) - G^\top \nabla \psi(y + G u + H w_z) , \\
        \eps \dot y &=  A y + \eps A^{-1}B \dot u + \eps A^{-1} E \dot w_z,\label{eq:ex-LTI-grad-shifted-a}
    \end{align}
\end{subequations}
with initial conditions $u(0) = u_0$ and $y(0) = z_0 - \subscr{z}{eq}(u_0, w_z(0)).$
On the other hand, the reduced model is, 
    \begin{align}\label{eq:ex-LTI-grad-reduced}
        \dot u_{\text{r}} &= - \nabla \phi(\subscr{u}{r}) - G^\top \nabla \psi(G \subscr{u}{r} + H w_z),
        \end{align}
with initial condition $\subscr{u}{r}(0) = u(0) = u_{0}$.  
Here, we emphasize that~\eqref{eq:ex-LTI-grad-reduced} is exactly the open-loop gradient flow~\eqref{eq:ex-LTI-openloop}. Moreover, the equilibrium trajectory $\subscr{u}{eq,r}(w_z(t))$ of~\eqref{eq:ex-LTI-grad-reduced} coincides with the unique optimal solution $u^*(w_z(t))$ of the problem~\eqref{eq:opt-prob-1} (note that we write these as functions of $w_z(t)$ since these are dependent on $w_z(t)$). In Theorem~\ref{th:ex-LTI-grad-eq-tracking}, we provide an equilibrium tracking bound with respect to $\subscr{u}{eq,r}(w_z(t)) \equiv u^*(w_z(t))$. 
For ease of reference, Table~\ref{tab:table-of-Lips-symbols-LTI-grad} documents relevant (one-sided) Lipschitz constants for the closed-loop system~\eqref{eq:ex-LTI-grad}.

% \begin{table}[t!]
%     \centering
%     \begin{tabular}{p{3.3cm}|p{3.3cm}} 
%  \hline \hline
% { \vspace{-.4cm}
% \begin{align*}
%     c_f &= \nu \\
%  \ell_{f_z} &= \Lip_z(f)  \\
%  \ell_{f_\eps} &= \ell_{\psi}\|G\| \\
%  \ell_{g_w} &= \|E\|  
% \end{align*}\vspace{-.4cm}} & {\vspace{-.4cm} \begin{align*}
%     c_g &= -\mu(A) \\
%  \ell_{g_x} &= \|B\|  \\
%  \ell_{g_\eps} &= \Lip_\eps (g) \\
%  \ell_{f_{w_z}} &= \|G\| \ell_\psi \|H\|
% \end{align*} \vspace{-.4cm}}  \\
%  \hline \hline
%     \end{tabular}
%     \caption{Constants for Closed Loop System~\eqref{eq:ex-LTI-grad} \vspace{-.4cm}}
%     \label{tab:table-of-Lips-symbols-LTI-grad}
% \end{table}

\begin{table}[t!]
    \centering
    \begin{tabular}{p{2.8cm}|p{2.8cm}|p{2.8cm}|p{2.8cm}} 
 \hline \hline
{ \vspace{-.4cm}
\begin{align*}
    c_f &= \nu \\
 \ell_{f_z} &= \Lip_z(f) % \\
 % \ell_{f_\eps} &= \ell_{\psi}\|G\| \\
 % \ell_{g_w} &= \|E\|  
\end{align*}\vspace{-.4cm}} & { \vspace{-.4cm}
\begin{align*}
 %    c_f &= \nu \\
 % \ell_{f_z} &= \Lip_z(f)  \\
 \ell_{f_\eps} &= \ell_{\psi}\|G\| \\
 \ell_{g_w} &= \|E\|  
\end{align*}\vspace{-.4cm}} & { \vspace{-.4cm}
\begin{align*}
   c_g &= -\mu(A) \\
 \ell_{g_x} &= \|B\|  % \\
 % \ell_{g_\eps} &= \Lip_\eps (g) \\
 % \ell_{f_{w_z}} &= \|G\| \ell_\psi \|H\|
\end{align*}\vspace{-.4cm}}  & {\vspace{-.4cm} \begin{align*}
 %    c_g &= -\mu(A) \\
 % \ell_{g_x} &= \|B\|  \\
 \ell_{g_\eps} &= \Lip_\eps (g) \\
 \ell_{f_{w_z}} &= \|G\| \ell_\psi \|H\|
\end{align*} \vspace{-.4cm}}  \\
 \hline \hline
    \end{tabular}
    \caption{Constants for Closed Loop System~\eqref{eq:ex-LTI-grad} \vspace{-.4cm}}
    \label{tab:table-of-Lips-symbols-LTI-grad}
\end{table}

\begin{theorem}[Equilibrium Tracking]\label{th:ex-LTI-grad-eq-tracking}
    Consider the closed loop system~\eqref{eq:ex-LTI-grad} 
    with $A \in \R^{n_x} \times \R^{n_x}$ Hurwitz, $t \mapsto w_z(t)$ locally absolutely continuous, and Assumptions~\ref{A:Lipschitz-costs}-\ref{A:strongly-convex-cost} satisfied. Assume that, 
    \begin{align}
        \eps < \eps_{0,\text{ofo}}^* := \frac{-\mu(A)}{\|A^{-1} BG^\top\| \ell_{\psi}}.
    \end{align}
    Then, the closed loop system~\eqref{eq:ex-LTI-grad} has a unique solution $z(t), u(t)$, for $t \geq 0$, that satisfies 
    \begin{subequations}\label{eq:ex-lti-grad-close-soln}
\begin{align}\label{eq:ex-lti-grad-close-soln-b}  
    \|u(t) - u^*(w_z(t))\| &\leq \eps \frac{\|G\|\ell_{\psi}\|H\|}{\nu(-\mu(A) - \eps \|A^{-1} B G^\top \|\ell_{\psi})}\bar w_z \\
    & \quad\quad + E_{u,\text{ofo}}(t) + \frac{\|G\|\ell_{\psi}\|H\|}{\nu^2}\bar w_z , \nonumber \\
    \label{eq:ex-lti-grad-close-soln-a} 
    \|z(t) - \subscr{z}{eq}(u^*(w_z(t)),w_z(t))\| &\leq  \eps \frac{\|H\| \bar w_z}{-\mu(A) - \eps \|A^{-1} B G^\top \|\ell_{\psi}} \Big(1 + \frac{\|G\|\ell_{\psi}\|H\|}{\nu}\Big)\\
    & \quad\quad + E_{z,\text{ofo}}(t), \nonumber
\end{align}
\end{subequations}
%\textcolor{red}{\eqref{eq:ex-lti-grad-close-soln-a} doesn't fit in one line, hence putting then both on two lines... }
where $\bar w_z := \operatorname{ess} \sup_{\tau \in [0, t]} \|\dot w_z(\tau)\|$, 
$t \mapsto u^*(w_z(t))$ is the optimizer of the optimization problem~\eqref{eq:opt-prob-1},
and the maps $t \mapsto E_{u,\text{ofo}}(t)$ and $t \mapsto E_{z,\text{ofo}}(t)$ tend to $0$ as $t \to +\infty$. \hfill  \QEDB 
\end{theorem}

We omit the expressions of $E_{u,\text{ofo}}(t)$ and $E_{z,\text{ofo}}(t)$  due to space limitations. To prove~\eqref{eq:ex-lti-grad-close-soln-b}, we first apply the triangle inequality on $\|u(t) - u^*(w_z(t))\| \leq \|u(t) - u_\text{r}(t)\| + \|\subscr{u}{r}(t) - u^*(w_z(t))\|$. We bound the first term using our result in Theorem~\ref{thm:closeness-of-solns-general} and bound the second term using~\cite[Thm. 2]{AD-VC-AG-GR-FB:23f}, recalling that $u^*(w_z(t)) \equiv \subscr{u}{eq,r}(w_z(t))$.  We showed that $z(t)$ tends to its steady state corresponding to the optimizer $u^*(w_z)$ by computing $\|z - \subscr{z}{eq}(u^*(w_z), w_z)\| \leq \|z - \subscr{z}{eq}(\subscr{u}{r}, w_z)\| + \|\subscr{z}{eq}(\subscr{u}{r}, w_z) - \subscr{z}{eq}(u^*(w_z), w_z)\|$, 
where the first term is bounded using our Theorem~\ref{thm:closeness-of-solns-general} and $\|\subscr{z}{eq}(\subscr{u}{r},w_z) - \subscr{z}{eq}(u^*(w_z), w_z))\| \leq \|A^{-1} B\| \|\subscr{u}{r} - u^*(w_z)\|$, which is then bounded using~\cite[Thm. 2]{AD-VC-AG-GR-FB:23f}. 

It is interesting to notice that $u(t)$ approaches the trajectory of the optimizer $u^*(w_z(t))$ up to a neighborhood depending on the time-variability of the disturbance $w_z(t)$; a similar behavior 
can be seen for $z(t)$. 
These results are in line with the input-to-state stability (ISS) results of~\cite{MC-ED-AB:20,LC-GB-EDA:22}.

\vspace{.1cm}

\begin{remark}[Differences with Lyapunov Theory]\label{rem:bound-eta-larger-for-ct}
Compared to the Lyapunov-based stability analysis of systems of the form~\eqref{eq:ex-LTI-grad},
our contraction-theoretic approach  allows  one to derive transient and asymptotic bounds on the fast and slow variables separately in~\eqref{eq:ex-lti-grad-close-soln} instead of one bound on the concatenated state.    
\hfill $\triangle$
\end{remark}

\section{Linear Time-Invariant Dynamics}
\label{sec:LTI-dynamics}

We consider the \textit{two-time scale LTI system},
\begin{subequations}\label{eq:LTI-dynamics}
\begin{align}\label{eq:LTI-dynamics-slow}
    \dot x &= A x + B z, \hspace{1.0cm} x(0) = x_0 \\
    \label{eq:LTI-dynamics-fast}
    \epsilon \dot z &= C x + D z,  \hspace{1.0cm} z(0) = z_0
\end{align}
\end{subequations}
where the states of the two sub-systems are $x \in \R^{n_x}$ and $z \in
\R^{n_z}$, the matrices are $A\in\R^{n_x} \times \R^{n_x}$, $B\in\R^{n_x}
\times \R^{n_z}$, $C\in\R^{n_z}\times \R^{n_x}$, and $D\in\R^{n_z} \times
\R^{n_z}$, and $\eps$ is a parameter.

%% Assume that there exist norms $\|\cdot\|_x$ and $\|\cdot\|_z$
%% with compatible weak pairings $\llbracket \cdot ; \cdot \rrbracket_x$ and
%% $\llbracket \cdot ; \cdot \rrbracket_z$ on $\R^{n_x}$ and $\R^{n_z}$,
%% respectively; for notational simplicity, let $\|\cdot\|_x$ and
%% $\|\cdot\|_z$ denote the induced matrix norms on $\R^{n_x \times n_x}$ and
%% $\R^{n_z \times n_z}$, respectively.

First, we assume that the matrix $D$ is Hurwitz and select a norm
$\|\cdot\|_z$ on $\R^{n_z}$ such that $\mu_z(D)$ is negative and within
$\eps$ of the spectral abscissa of $D$, for an arbitrarily small $\eps$
(e.g., see~\cite[Lemma~2.27]{FB:24-CTDS}). Accordingly, define $c_g :=
-\mu_z(D) \leq \lvert \alpha(D) \lvert +\eps$.  Note that
Assumption~\ref{as:aut-fast-str-contracting} is satisfied.

Since $D$ is Hurwitz, for any fixed $x \in \R^{n_x}$ the unique solution to
$0 = C x + D z^*(x)$ is given by $z^*(x) := -D^{-1} C x$. %; clearly, the map $x \mapsto z^*(x)$ satisfies Assumption~\ref{as:extra-regularity_auto}. 
As
in the previous sections, consider the change of variables $y:= z - z^*(x)
= z + D^{-1} C x$, and write the two-time scale LTI
system~\eqref{eq:LTI-dynamics} in the new variables to define the
\textit{shifted LTI system}, given by:
\begin{subequations}\label{eq:LTI-dynamics-new-var}
\begin{align}\label{eq:LTI-dynamics-slow-new-var}
    \dot x &= (A - B D^{-1} C) x + B y \\
    \label{eq:LTI-dynamics-fast-new-var}
    \epsilon \dot y &= (D + \eps D^{-1} C B) y + \eps D^{-1} C (A - B D^{-1} C) x
\end{align}
\end{subequations}
with initial conditions $x(0) = x_0$ and $y(0) = z_0 - z^*(x_0)$.
Moreover, the \textit{reduced LTI model} is,
\begin{align}\label{eq:LTI-dynamics-reduced}
    \dot x_{\text{r}} &= (A - B D^{-1} C) x_{\text{r}}, \hspace{1.0cm}
    x_{\text{r}}(0) = x_{{\text{r}},0} .
\end{align}
We now assume that the matrix $(A - BD^{-1} C) \in \R^{n_x \times n_x}$ is
Hurwitz and select a norm $\|\cdot\|_x$ on $\R^{n_x}$ such that $\mu_x(A -
BD^{-1} C)$ is negative and within $\eps$ of the spectral abscissa of $A -
BD^{-1} C$, for an arbitrarily small $\eps$ (e.g.,
see~\cite[Lemma~2.27]{FB:24-CTDS}). Accordingly, define $c_f := -\mu_x(A -
BD^{-1} C) \leq \lvert \alpha(A - BD^{-1} C) \lvert+\eps$.  Note that
Assumption~\ref{as:red-str-contracting-aut} is satisfied.

Finally, we need some more notation. Given matrices $F \in \R^{n_x \times
  n_z}$ and $F'\in \R^{n_z \times n_x}$, we define the induced norms:
\begin{align*}
    \|F\|_{z \to x} &:= \max\{ \|F z\|_x : z \in \R^{n_z}, ~ \|z\|_z = 1 \}, \\
    \|F'\|_{x \to z} &:= \max\{ \|F' z\|_z : x \in \R^{n_x}, ~ \|x\|_x = 1 \}.
\end{align*}
In the following, we discuss how LTI dynamics yield tighter bounds relative
to the ones in Theorem~\ref{thm:closeness-of-solns-aut} by following the same proof method of Section~\ref{sec:proofs} (rather than applying the result of Theorem~\ref{thm:closeness-of-solns-aut} outright). 
% alternative version of Theorem~\ref{thm:str-contractive-whole-system-aut}
% by exploiting key properties of the LTI systems.

\begin{theorem}[Closeness of the solutions for LTI systems]
\label{thm:closeness-of-solns-LTI}
Consider the two-time scale LTI system~\eqref{eq:LTI-dynamics}, assume that
the matrices $D$ and $(A - BD^{-1} C)$ are Hurwitz, and select norms
$\|\cdot\|_x$ and $\|\cdot\|_z$ as explained above.  Suppose that
\begin{align}
  0 < \eps < \eps_{0,\text{LTI}}^* := \frac{\lvert \mu_z(D)\lvert}{\|D^{-1} C  B\|_z}.
\end{align}
Then, the solutions $x(t), z(t)$, for $t \geq 0,$ satisfy the bounds
\allowdisplaybreaks
\begin{subequations}\label{eq:thm-close-solns-LTI}
  \begin{align}
     \|x(t) - \subscr{x}{r}(t)\|_x &\leq e^{\mu_x(A - BD^{-1}C) t} \|x(0) - \subscr{x}{r}(0)\|_x + G'_x(t) \label{eq:thm-close-solns-LTI-a}, \\
     \|z(t) - z^*(\subscr{x}{r}(t))\|_z &\leq G'_z(t),
    \label{eq:thm-close-solns-LTI-b} 
  \end{align}
\end{subequations}
where $t \mapsto \subscr{x}{r}(t)$ is a solution of the reduced LTI
model~\eqref{eq:LTI-dynamics-reduced} and the functions $t \mapsto G'_x(t)$
and $t \mapsto G'_z(t)$ satisfy $G'_x(t) \rightarrow 0$, $G'_z(t)
\rightarrow 0$ exponentially fast as $t \rightarrow +\infty$.  \QEDB
\end{theorem}

\vspace{.1cm}

The functions $G'_x(t)$ and $G'_z(t)$ are given as follows, where we additionally define $-\subscr{c}{L,y} := \mu_z(D) \eps^{-1} + \|D^{-1} C B\|_z$ and $\subscr{A}{red} := A - BD^{-1}C$ for brevity. 

\noindent \emph{Case~1}: $-\mu_x(\subscr{A}{red}) \neq \subscr{c}{L,y}$. 
\begingroup
\allowdisplaybreaks
\begin{align*}
    G'_x(t) &= \|B\|_{z \to x} \|z(0) + D^{-1} C x(0)\|_z \Big(\frac{e^{-\subscr{c}{L,y} t} - e^{\mu_x(\subscr{A}{red}) t}}{\mu_x(\subscr{A}{red}) -  \subscr{c}{L,y}} \Big) \\
    & \quad + \frac{\|B\|_{z \to x}\|D^{-1} C\|_{x \to z}\|\subscr{A}{red} x(0)\|_x }{\subscr{c}{L,y} - \mu_x(\subscr{A}{red})}\Big(t e^{\mu_x(\subscr{A}{red})t} \Big)\\
    & \quad + \frac{\|B\|_{z \to x}\|D^{-1} C\|_{x \to z}\|\subscr{A}{red} x(0)\|_x }{\subscr{c}{L,y} - \mu_x(\subscr{A}{red})}\Big( \frac{e^{\mu_x(\subscr{A}{red})t} - e^{-\subscr{c}{L,y} t}}{\subscr{c}{L,y} - \mu_x(\subscr{A}{red})} \Big),\\
    G'_z(t) &= e^{-\subscr{c}{L,y}t} \|z(0) + D^{-1} C x(0)\|_z  + \|D^{-1} C\|_{x \to z} e^{-\subscr{c}{L,y}t}  \\
    & \quad + \|D^{-1} C\|_{x \to z}\left( \|\subscr{A}{red} x(0)\|_x \Big( \frac{e^{\mu_x(\subscr{A}{red})t}  - e^{-\subscr{c}{L,y}t} }{\subscr{c}{L,y} - \mu_x(\subscr{A}{red})} \Big) + G'_x(t)\right).
\end{align*}
\endgroup

\noindent \emph{Case~2}: $-\mu_x(\subscr{A}{red}) = \subscr{c}{L,y}$.
\begingroup
\allowdisplaybreaks
\begin{align*}
    G'_x(t) & = \|B\|_{z \to x} \|z(0) + D^{-1} Cx(0)\|_x t e^{-\subscr{c}{L,y} t}  + \frac{1}{2}t^2 e^{\mu_x(\subscr{A}{red}) t},\\
     G'_z(t) &= e^{-\subscr{c}{L,y} t} \|z(0) + D^{-1} Cx(0)\|_x \\
    & \quad\quad + \|D^{-1} C\|_{x \to z} \big( e^{-\subscr{c}{L,y}t} + G'_x(t) +  \|\subscr{A}{red}x(0)\|_x t e^{-\subscr{c}{L,y}t} \big).
\end{align*}
\endgroup

We note that the expressions for $G'_x(t)$ and $G'_z(t)$ above are tighter than the ones that would have been obtained by substituting $\ell_{f_z} = \|B\|_{z \to x}$, $\ell_{g_x} = \|C\|_{x \to z}$, $c_g = -\mu(D)$, and $\ell_{f_x} = \|A\|_x$ in Theorem~\ref{thm:closeness-of-solns-aut}.  
For example, consider the following chain of bounds for the term $\|D^{-1} C\|_{x \to z}$ appearing in both functions $G'(x)$ and $G'(z)$: $\|D^{-1} C \|_{x \to z} \leq \|D^{-1}\|_z \|C\|_{x \to z} \leq \frac{\|C\|_{x \to z}}{-\mu(D)}$,
where $\|D^{-1}\|_z \leq -\frac{1}{\mu(D)}$ by the uniform monotonicity property~\cite[Lem. 2.11]{FB:24-CTDS}. 
This implies that $\frac{\ell_{g_x}}{c_g}$ is an upper bound on $\|D^{-1} C \|_{x \to z}$.  
% Similarly, we customize Theorem~\ref{thm:str-contractive-whole-system-aut} for the shifted LTI system~\eqref{eq:LTI-dynamics-new-var}.
Next, we offer a new contractivity result for the shifted LTI system~\eqref{eq:LTI-dynamics-new-var}.

\begin{theorem}[Contractivity of the shifted LTI system]\label{thm:str-contractive-whole-system-LTI}
  Consider the two-time scale LTI system~\eqref{eq:LTI-dynamics}, assume
  that the matrices $D$ and $(A - BD^{-1} C)$ are Hurwitz, and select norms
  $\|\cdot\|_x$ and $\|\cdot\|_z$ as explained above.
  %% Consider the shifted LTI system~\eqref{eq:LTI-dynamics-new-var} and
  %% assume that the matrices $D$ and $(A - BD^{-1} C)$ are Hurwitz.
Suppose $\eps$ satisfies $0 < \eps < \eps^*_\text{LTI}$, where
\begin{align*}
\eps^*_\text{LTI} := 
   |\lognorm{D}| \Bigl(\frac{\norm{B}_{z \to x} \|D^{-1}C (A - B D^{-1} C)\|_{x \to z}}
        {|\lognorm{A-BD^{-1}C}|}+ \|D^{-1}CB\|_z \Bigr)^{-1} .
\end{align*}
Then, the shifted LTI system~\eqref{eq:LTI-dynamics-new-var} is strongly
infinitesimally contracting with contraction rate $\lvert
\alpha(\Gamma_{\text{LTI},\eps}) \lvert$ with respect to the norm
$\|\cdot\|_{N_{\text{LTI}}}$, where
\begin{align} 
    & \Gamma_{\text{LTI},\eps} := \left[
    \begin{array}{c;{4pt/4pt}c}
        A - BD^{-1} C & B \\ 
        \hdashline[4pt/4pt]
        D^{-1} C(A -BD^{-1}C) & \eps^{-1} D + D^{-1} C B
    \end{array}
\right] \nonumber 
\end{align}
and $N_{\text{LTI}} \in \R_{>0}^2$ is defined in the proof. % for Theorem~\ref{thm:str-contractive-whole-system-general}.
\QEDB
\end{theorem}

\begin{proof}
\vspace{.1cm}

\noindent
    \textit{(a) Apply Lemma~\ref{lem:hurwitz-gain-matrix} to $\Gamma_{\text{LTI},\eps}$.} % As a next step, w
For this step, we use Lemma~\ref{lem:hurwitz-gain-matrix} (provided in the Appendix) to find sufficient conditions on $\eps$ such that  $\Gamma_{\text{LTI},\eps}$ is Hurwitz.  

\vspace{.1cm}

\noindent
\textit{(b) Obtain Contraction Rate.} When  $\Gamma_{\text{LTI},\eps}$ is Hurwitz, we can apply~\cite[Th. 3.23]{FB:24-CTDS} to find the contraction rate. 
First, 
consider the $\ell_2$-norm. 
Define $(v_1, v_2), (w_1, w_2) \in \R^2_{>0}$ to be the right and left dominant eigenvectors of $\Gamma_{\text{LTI},\eps}$, 
and let
$N_{\text{LTI}} := (N_1, N_2) := \Big( \frac{w_1^{1/2}}{v_1^{1/2}}, \frac{w_2^{1/2}}{v_2^{1/2}} \Big) \in \R^2_{>0}$ so that $\bar N := \operatorname{diag}(N_1,N_2) \in \R^{2 \times 2}.$ 
Note that since $\Gamma_{\text{LTI},\eps} \in \R^{2 \times 2}$ is Hurwitz and Metzler with off-diagonal etries strictly positive, $\Gamma_\eps$ is irreducible.
By application of~\cite[Lem. 2.31]{FB:24-CTDS}, a logarithmically optimal norm for $\Gamma_{\text{LTI},\eps}$ is $\|\cdot\|_{2, \bar N^{1/2}}$. 
Using $N_{\text{LTI}} := (N_1, N_2)$, we then define the weighted $\ell_2$ composite norm $\|\cdot\|_N$ for the concatenated state $(x,z) \in \R^{n_x + n_z}$ by, 
\[\|(x, z)\|_N^2 = N_1 \|x\|_z^2 + N_2 \|z\|_z^2.\]
Finally, by~\cite[Th. 3.23]{FB:24-CTDS},  
the shifted system~\eqref{eq:gen-new-var-w} is strongly infinitesimally contracting with respect to $\|\cdot\|_{N}$ with rate $\lvert \alpha(\Gamma_{\text{LTI},\eps}) \lvert$.
\end{proof}

The following remark emphasizes differences relative to classical results for the interconnection of LTI systems.

\begin{remark}[Comparison with 
 Network Contraction Theorem]\label{rem:interconnection-theorem} A classical
  approach for the analysis of the two-time scale LTI system~\eqref{eq:LTI-dynamics} relies on
  the Interconnection
  Theorem~\cite{GR-MDB-EDS:13},\cite[Thm. 3.23]{FB:24-CTDS}. The
  Network Contraction Theorem is based on the assumptions $\mu(A) < 0, \mu(D) <
  0$ and $\mu(A) \mu(D) > \|B\|_{z \to x}\|C\|_{x \to z}$. Although the Network Contraction Theorem
  does not require a parameter $\eps$ to satisfy given bounds, assuming
  that the matrices $A$ and $D$ are Hurwitz is a stronger requirement than
  our assumptions in Theorem~\ref{thm:str-contractive-whole-system-LTI}. In
  fact, defining the block matrix associated with the two-time scale LTI system \eqref{eq:LTI-dynamics}
  as,
 \begin{align}
   \mcA_\eps := \begin{bmatrix} \eps A & \eps B \\ C & D
   \end{bmatrix}, \label{eq:blockmatrix}
 \end{align}
the following implications hold: 
\begin{equation} \label{diagram:scaled-block}
   \begin{tikzcd}
     \parbox[t]{.4\textwidth}
            {\centering $\lognorm{A}<0$, $\lognorm{D}<0$, 
       and $\lognorm{A}\lognorm{D}> \norm{B}_{z \to x} \norm{C}_{x \to z}$ }
     \arrow[d,Rightarrow]  \arrow[r,Rightarrow, shorten <= 5pt, shorten >= 5pt] & 
           \parbox[t]{.35\textwidth}{\centering \(\mcA_\eps\) is Hurwitz  \(\forall~\eps>0\) \arrow[d,Rightarrow]}
           \\
             \parbox[t]{.4\textwidth}{\centering \(D\)  and \( A - B D^{-1} C \) are Hurwitz }
             \arrow[r,Rightarrow, shorten <= 3pt, shorten >= 3pt]
             & \parbox[t]{.5\textwidth}{\centering \(\exists ~  \eps_\text{LTI}^* > 0 \) s.t. \(\mcA_\eps\) is Hurwitz \(\forall ~ \eps<\eps_\text{LTI}^*\)} 
   \end{tikzcd}  
 \end{equation}
where $\eps^*$ is precisely given in
Theorem~\ref{thm:str-contractive-whole-system-LTI}. The proof of
diagram~\eqref{diagram:scaled-block} is provided in the Appendix. Our
approach is well-suited for applications where the matrix $A$ is not
Hurwitz, including our examples in
Section~\ref{sec:extensions-and-applications}.  \hfill $\triangle$
\end{remark}

\section{Conclusion}
\label{sec:conclusions}

In this work, we proposed a novel method to systematically analyze
singularly perturbed systems via contraction theory and derived new
equilibrium tracking bounds for systems within the OFO framework. By
assuming that (i) the fast subsystem and the reduced model are strongly
infinitesimally contractive, (ii) the interconnections are Lipschitz, and
(iii) the time-scale inducing parameter $\epsilon$ is suitably bounded, we
guarantee that the solutions of a two-time scale system approach the
solutions of its reduced model, as shown in
Theorem~\ref{thm:closeness-of-solns-general}. This is achieved by
explicitly computing bounds on their differences.  The bounds on $\epsilon$
are directly related to the Lipschitz constants and the properties of the
disturbances. Our results leverage assumptions that are weaker than those
used in prior works exploring the use of contraction theory, therefore
making them applicable to a wider array of systems.  We demonstrated that
our contraction-theoretic method is applicable to the OFO setup without
altering the typical assumptions imposed in traditional OFO
settings. Specifically, we showed that these common assumptions not only
ensure exponential stability but also contractivity.  Finally, we obtained
novel results for two-time scale LTI systems using the analysis techniques
of our proposed contraction-theoretic method.

In the future, we plan to investigate (i) weakly contracting dynamical
systems, (ii) contraction-theoretic methods for results for two-time scale
systems modeled by stochastic differential equations, and (iii) the
application of our methods to the analysis and control of neural networks.

\section*{Appendix}

\subsection{Auxiliary results for Hurwitz matrices}

\begin{lemma}
\label{lem:hurwitz-gain-matrix}
Let $a_{11},a_{12},a_{21},a_{22},d_{22} > 0$ and $d_{11},d_{21} \geq 0$. Let $\eps > 0$, and define the matrix, 
\begin{align}\label{eq:template-gain-matrix}
    \mathcal{G}_\eps = \begin{bmatrix}
        -a_{11} + d_{11} \eps & a_{12} \\
        a_{21} + d_{21} & -\eps^{-1} a_{22} + d_{22}
    \end{bmatrix}.
\end{align}
The if $\eps$ satisfying 
\begin{align}\label{eq:eps-conditions}
    0 < \eps < \min \left\{ \frac{a_{11}}{d_{11}}, \frac{a_{22}}{d_{22}}, \frac{a_{11} a_{22}}{a_{12}(a_{21} + d_{21} + d_{22})}  \right\}, 
\end{align}
then the matrix $\mathcal{G}_\eps$ is Hurwitz. 
\end{lemma}

\begin{proof}
    To ensure that $\mathcal{G}_\eps$ is Hurwitz, we require that the diagonal elements are all negative and that $\det(\mathcal{G}_\eps) > 0$. To satisfy the first requirement, 
    % we must ensure that:
    % \begin{align*}
    %     -a_{11} + d_{11} \eps < 0 \text{ and } -\eps^{-1}a_{22} + d_{22} < 0,
    % \end{align*}
    $\eps$ must be bounded as means that 
    \begin{align}
        \eps < \frac{a_{11}}{d_{11}} \text{ and } \eps < \frac{a_{22}}{d_{22}} \, .
    \end{align}
    This gives the first two entries of~\eqref{eq:eps-conditions}. 
    Next, we find conditions such that $\det(\mathcal{G}_\eps) < 0$. Accordingly, we have that: 
\begin{align*}
   %0 (-a_{11} + d_{11} \eps) (-\eps^{-1} a_{22} + d_{22}) &> (a_{21} +  d_{21}c_0) a_{12} \\
    \eps^{-1} a_{11} a_{22} - a_{22} d_{11} - a_{11} d_{22} + \eps d_{11} d_{22} > a_{21} a_{12} +  a_{12} d_{21}. 
\end{align*}
We already require that $\eps < a_{22} / d_{22}$. We use this to derive the following inequality
% \begin{align*}
%     \eps^{-1} a_{11} &a_{22} - a_{22} d_{11} - a_{11} d_{22} + \underbrace{\frac{a_{22}}{d_{22}}}_{\eps < \frac{a_{22}}{d_{22}}} d_{11} d_{22} > \\
%     &\eps^{-1} a_{11} a_{22} - a_{22} d_{11} - a_{11} d_{22} + \eps d_{11} d_{22}.
% \end{align*}
    \begin{align*}
        \eps^{-1} a_{11} a_{22} - a_{22} d_{11} - a_{11} d_{22} + \frac{a_{22}}{d_{22}} d_{11} d_{22} > a_{21} a_{12} +  a_{12} d_{21} 
    \end{align*}
and, simplifying, 
\begin{align}
    %\eps^{-1} a_{11}& a_{22} - a_{22} d_{11} - a_{11} d_{22} + \underbrace{\frac{a_{22}}{d_{22}} d_{11} d_{22}}_{a_{22} d_{11}} \\
    %    & ~~~~~~~~~~~~~~~~~~~~~~~~~~~~ > a_{21} a_{12} +  a_{12} d_{21}\\
    %\implies &
    \eps^{-1} a_{11} a_{22} - a_{11} d_{22} > a_{21} a_{12} +  a_{12} d_{21}. 
\end{align}
We therefore obtain the following bound for  $\eps$: 
\begin{align*}
    %\eps^{-1} a_{11} a_{22} &- a_{11} d_{22} > a_{21} a_{12} + a_{12} d_{21} c_0 \\
    %\implies 
    \eps < \frac{a_{11} a_{22}}{a_{21} a_{12} + a_{12} d_{21}  + a_{12} d_{22}}.
\end{align*}
which gives the third condition in~\eqref{eq:eps-conditions}.

\end{proof}

%% NOTE: Latex threw warnings when we had the old subsection title; can keep, but it doesn't like the \eqref{} in the subsection{}.
% \subsection{Proof of the properties~\eqref{diagram:scaled-block}}
\subsection{Proof of Commutative Diagram}
\label{sec:appendix-FB-exercise}
 
Here, we prove the four properties in~\eqref{diagram:scaled-block} 
for the block matrix~\eqref{eq:blockmatrix}. 
To this end, it is convenient to denote the four properties in the
diagram~\eqref{diagram:scaled-block} by $P(i,j)$, for $i,j\in\{1,2\}$.

% \vspace{.1cm}
 
\noindent $P(1,2) \implies P(2,2)$.  This implication obvious. 

% \vspace{.1cm}

\noindent $P(1,1) \implies P(1,2)$ This follows 
 from~\cite[Thm. 2.13]{FB:24-CTDS} and~\cite[Corollary 2.43]{FB:24-CTDS}.

\vspace{.1cm}

\noindent $P(1,1) \implies P(2,1)$. Clearly $D$ is
 Hurwitz and it suffices to show that $A-BD^{-1}C$ is Hurwitz. We compute
 \begin{align*}
   \lognorm{A-BD^{-1}C} &\leq 
   \lognorm{A} +
   \lognorm{-BD^{-1}C}
   \leq
   \lognorm{A}+\norm{BD^{-1}C} \\
    & \leq
   \lognorm{A} +  \norm{B} \norm{D^{-1}} \norm{C} 
   \leq
   \lognorm{A} -  \frac{ \norm{B}
     \norm{C} }{ \lognorm{D} }  \\  
   &= - \frac{1}{|\lognorm{D}|} \Big(
   |\lognorm{A}|  |\lognorm{D}| - \norm{B}  \norm{C}
   \Big) <0,
 \end{align*}
 where the last inequality follows from the assumption. Therefore,
 $A-BD^{-1}C$ is Hurwitz.

\noindent $P(2,1) \implies P(2,2)$. We first prove that
 $\mcA_\eps$ is similar to
 \begin{equation}
 \begin{split}
   \mcB_\eps &= \begin{bmatrix}
   A - B D^{-1} C & B\\
   D^{-1}C (A - B D^{-1} C) & \epsilon^{-1} D + D^{-1} C B
   \end{bmatrix}  \\
   &= \eps^{-1}
   \begin{bmatrix} I_n & 0 \\ D^{-1}C & I_m \end{bmatrix}
   \mcA_\eps
   \begin{bmatrix} I_n & 0 \\ - D^{-1}C & I_m \end{bmatrix},
 \end{split}
 \end{equation}
 by noting that 
 $$ 
 \begin{bmatrix} I_n & 0 \\ D^{-1}C &
   I_m \end{bmatrix} \begin{bmatrix} I_n & 0 \\ - D^{-1}C &
   I_m \end{bmatrix}=I_{n+m}.
   $$  
   Finally, an application of\cite[Thm. 2.13]{FB:24-CTDS} implies the Hurwitzness of $\mcB_\eps$. Specifically, given norms on $\R^n$ and
 $\R^m$ such that $\lognorm{A-BD^{-1}C}<0$ and $\lognorm{D}<0$,
 the matrix $\mcB_\eps$ is Hurwitz if
 \begin{gather}
 \begin{split}
   \label{eq:messy-composite-singular}
   |\lognorm{A-BD^{-1}C}| \cdot & |\lognorm{\eps^{-1}D+D^{-1} C B}| > \norm{B}  \norm{D^{-1}C (A - B D^{-1} C)}.
   \end{split}
 \end{gather}  
 For $\eps<\eps^*_1:=|\lognorm{D}|/\norm{D^{-1}CB}$, we note that
 $\eps^{-1}>\norm{D^{-1}CB}/|\lognorm{D}|$ so that we obtain
 \begin{equation*}
   \lognorm{\eps^{-1}D+D^{-1}CB} \leq
   \eps^{-1}\lognorm{D}+\norm{D^{-1}CB}<0.
 \end{equation*}
 Therefore, we note that $$|\lognorm{\eps^{-1}D+D^{-1}CB}| \geq
 \eps^{-1}|\lognorm{D}|-\norm{D^{-1}CB}$$ and that a sufficient condition
 for the inequality~\eqref{eq:messy-composite-singular} is
 \begin{align*}
   \label{eq:messy-composite-singular:2}
    |\lognorm{A-B & D^{-1}C}| \cdot \big(\eps^{-1}|\lognorm{D}|-\norm{D^{-1}CB}\big)
   > \norm{B}  \norm{D^{-1}C (A - B D^{-1} C)}
   \\
   &\iff
   \eps^{-1}|\lognorm{D}|
>
   \frac{\norm{B} \norm{D^{-1}C (A - B D^{-1} C)}}
        {|\lognorm{A-BD^{-1}C}|}  +\norm{D^{-1}CB}
   \\
   &\iff
   \eps
   < \eps_2^* := 
   |\lognorm{D}|
   \Big(\frac{\norm{B} \norm{D^{-1}C (A - B D^{-1} C)}}
        {|\lognorm{A-BD^{-1}C}|} +\norm{D^{-1}CB}\Big)^{-1}.
 \end{align*}
It is easy to
 see that $\eps_2^*<\eps_1^*$. This completes the proof of the four implications. 

\bibliographystyle{siamplain}
%% \bibliography{alias,biblio}
\bibliography{alias,Main,FB}

\begin{thebibliography}{10}

\bibitem{MAAR-DA-EDS:23}
{\sc M.~A. Al-Radhawi, D.~Angeli, and E.~D. Sontag}, {\em On structural
  contraction of biological interaction networks}, arXiv preprint
  arXiv:2307.13678,  (2023), \url{https://doi.org/10.48550/arXiv.2307.13678}.

\bibitem{ZA-EDS:14b}
{\sc Z.~Aminzare and E.~D. Sontag}, {\em Contraction methods for nonlinear
  systems: {A} brief introduction and some open problems}, in {IEEE} Conf.\ on
  Decision and Control, Dec. 2014, pp.~3835--3847,
  \url{https://doi.org/10.1109/CDC.2014.7039986}.

\bibitem{GB-JC-JIP-EDA:22}
{\sc G.~Bianchin, J.~Cort\'es, J.~I. Poveda, and E.~{Dall'Anese}}, {\em
  Time-varying optimization of {LTI} systems via projected primal-dual gradient
  flows}, IEEE Transactions on Control of Network Systems, 9 (2022),
  pp.~474--486, \url{https://doi.org/10.1109/TCNS.2021.3112762}.

\bibitem{bianchin2021planning}
{\sc G.~Bianchin, E.~Dall'Anese, J.~I. Poveda, D.~Jacobson, E.~J. Carlton, and
  A.~G. Buchwald}, {\em Planning a return to normal after the {COVID}-19
  pandemic: Identifying safe contact levels via online optimization}, arXiv
  preprint arXiv:2109.06025,  (2021),
  \url{https://doi.org/10.48550/arXiv.2109.06025}.

\bibitem{FB:24-CTDS}
{\sc F.~Bullo}, {\em Contraction Theory for Dynamical Systems}, Kindle Direct
  Publishing, {1.2}~ed., 2024, \url{https://fbullo.github.io/ctds}.

\bibitem{chow1990singular}
{\sc J.~Chow, J.~Winkelman, M.~Pai, and P.~Sauer}, {\em Singular perturbation
  analysis of large-scale power systems}, International Journal of Electrical
  Power \& Energy Systems, 12 (1990), pp.~117--126,
  \url{https://doi.org/10.1016/0142-0615(90)90007-X}.

\bibitem{MC-ED-AB:20}
{\sc M.~Colombino, E.~{Dall'Anese}, and A.~Bernstein}, {\em Online optimization
  as a feedback controller: {S}tability and tracking}, IEEE Transactions on
  Control of Network Systems, 7 (2020), pp.~422--432,
  \url{https://doi.org/10.1109/TCNS.2019.2906916}.

\bibitem{LC-GB-EDA:22}
{\sc L.~Cothren, G.~Bianchin, and E.~{Dall'Anese}}, {\em Online optimization of
  dynamical systems with deep learning perception}, IEEE Open Journal of
  Control Systems, 1 (2022), pp.~306--321,
  \url{https://doi.org/10.1109/OJCSYS.2022.3205871}.

\bibitem{AD-VC-AG-GR-FB:23f}
{\sc A.~Davydov, V.~Centorrino, A.~Gokhale, G.~Russo, and F.~Bullo}, {\em
  Time-varying convex optimization: A contraction and equilibrium tracking
  approach}, IEEE Transactions on Automatic Control,  (2023),
  \url{https://doi.org/10.48550/arXiv.2305.15595}.
\newblock Submitted.

\bibitem{AD-AVP-FB:22q}
{\sc A.~Davydov, A.~V. Proskurnikov, and F.~Bullo}, {\em {Non-Euclidean}
  contraction analysis of continuous-time neural networks}, IEEE Transactions
  on Automatic Control,  (2024),
  \url{https://doi.org/10.1109/TAC.2024.3422217}.
\newblock To appear.

\bibitem{dean2020robust}
{\sc S.~Dean, N.~Matni, B.~Recht, and V.~Ye}, {\em Robust guarantees for
  perception-based control}, in Learning for Dynamics and Control, PMLR, 2020,
  pp.~350--360, \url{https://doi.org/10.48550/arXiv.1907.03680}.

\bibitem{dean2021certainty}
{\sc S.~Dean and B.~Recht}, {\em Certainty equivalent perception-based
  control}, in Learning for Dynamics and Control, PMLR, 2021, pp.~399--411,
  \url{https://doi.org/10.48550/arXiv.2008.12332}.

\bibitem{DDC:13}
{\sc D.~{Del~Vecchio}}, {\em A control theoretic framework for modular analysis
  and design of biomolecular networks}, Annual Reviews in Control, 37 (2013),
  pp.~333--345, \url{https://doi.org/10.1016/j.arcontrol.2013.09.011}.

\bibitem{DDV-JJES:11}
{\sc D.~{Del~Vecchio} and J.-J.~E. Slotine}, {\em A contraction theory approach
  to singularly perturbed systems with application to retroactivity
  attenuation}, in {IEEE} Conf.\ on Decision and Control and European Control
  Conference, 2011, pp.~5831--5836,
  \url{https://doi.org/10.1109/CDC.2011.6160340}.

\bibitem{DDV-JJES:13}
{\sc D.~{Del Vecchio} and J.-J.~E. Slotine}, {\em A contraction theory approach
  to singularly perturbed systems}, IEEE Transactions on Automatic Control, 58
  (2013), pp.~752--757, \url{https://doi.org/10.1109/TAC.2012.2211444}.

\bibitem{YF-TGK:96}
{\sc Y.~Fang and T.~G. Kincaid}, {\em Stability analysis of dynamical neural
  networks}, IEEE Transactions on Neural Networks, 7 (1996), pp.~996--1006,
  \url{https://doi.org/10.1109/72.508941}.

\bibitem{fiez2019convergence}
{\sc T.~Fiez, B.~Chasnov, and L.~Ratliff}, {\em Implicit learning dynamics in
  {Stackelberg} games: {Equilibria} characterization, convergence analysis, and
  empirical study}, in International Conference on Machine Learning, 2020,
  pp.~3133--3144.

\bibitem{MG-VA-ST-DA:23}
{\sc M.~Giaccagli, V.~Andrieu, S.~Tarbouriech, and D.~Astolfi}, {\em {LMI}
  conditions for contraction, integral action, and output feedback
  stabilization for a class of nonlinear systems}, Automatica, 154 (2023),
  p.~111106, \url{https://doi.org/10.1016/j.automatica.2023.111106}.

\bibitem{AH-SB-GH-FD:21}
{\sc A.~Hauswirth, S.~Bolognani, G.~Hug, and F.~D{\"o}rfler}, {\em Timescale
  separation in autonomous optimization}, IEEE Transactions on Automatic
  Control, 66 (2021), pp.~611--624,
  \url{https://doi.org/10.1109/tac.2020.2989274}.

\bibitem{AH-ZH-SB-GH-FD:24}
{\sc A.~Hauswirth, Z.~He, S.~Bolognani, G.~Hug, and F.~D\"orfler}, {\em
  Optimization algorithms as robust feedback controllers}, Annual Reviews in
  Control, 57 (2024), p.~100941,
  \url{https://doi.org/10.1016/j.arcontrol.2024.100941}.

\bibitem{FCH:66}
{\sc F.~C. Hoppensteadt}, {\em Singular perturbations on the infinite
  interval}, Transactions of the American Mathematical Society, 123 (1966),
  pp.~521--535, \url{https://doi.org/10.2307/1994672}.

\bibitem{JJ-TIF:10}
{\sc J.~Jouffroy and T.~I. Fossen}, {\em A tutorial on incremental stability
  analysis using contraction theory}, Modeling, Identification and Control, 31
  (2010), pp.~93--106, \url{https://doi.org/10.4173/mic.2010.3.2}.

\bibitem{HKK:02}
{\sc H.~K. Khalil}, {\em Nonlinear Systems}, Prentice Hall, 3~ed., 2002.

\bibitem{JK-EAC:17}
{\sc J.~Kim and E.~A. Croft}, {\em Full-state tracking control for flexible
  joint robots with singular perturbation techniques}, IEEE Transactions on
  Control Systems Technology, 27 (2017), pp.~63--73,
  \url{https://doi.org/10.1109/TCST.2017.2756962}.

\bibitem{PK-PS:68}
{\sc P.~Kokotovi{\'c} and P.~Sannuti}, {\em Singular perturbation method for
  reducing the model order in optimal control design}, IEEE Transactions on
  Automatic Control, 13 (1968), pp.~377--384,
  \url{https://doi.org/10.1109/TAC.1968.1098927}.

\bibitem{PVK-HKK-JOR:99}
{\sc P.~V. Kokotovi\'c, H.~K. Khalil, and J.~{O'Reilly}}, {\em Singular
  Perturbation Methods in Control: Analysis and Design}, SIAM, 1999,
  \url{https://doi.org/10.1137/1.9781611971118}.

\bibitem{kokotovic1976singular}
{\sc P.~V. Kokotovi{\'c}, R.~E. O'Malley~Jr, and P.~Sannuti}, {\em Singular
  perturbations and order reduction in control theory — an overview},
  Automatica, 12 (1976), pp.~123--132,
  \url{https://doi.org/10.1016/S1474-6670(17)67778-4}.

\bibitem{LK-ME-JJES:22}
{\sc L.~Kozachkov, M.~Ennis, and J.-J.~E. Slotine}, {\em {RNNs} of {RNNs}:
  {Recursive} construction of stable assemblies of recurrent neural networks},
  in Advances in Neural Information Processing Systems, Dec. 2022,
  \url{https://doi.org/10.48550/arXiv.2106.08928}.

\bibitem{RJK-WHM-CM:16}
{\sc R.~J. Kutadinata, W.~H. Moase, and C.~Manzie}, {\em Extremum-seeking in
  singularly perturbed hybrid systems}, IEEE Transactions on Automatic Control,
  62 (2016), pp.~3014--3020, \url{https://doi.org/10.1109/TAC.2016.2607282}.

\bibitem{LSPL-ZEN-EM-JWSP:18}
{\sc L.~S.~P. Lawrence, Z.~E. Nelson, E.~Mallada, and J.~W. Simpson-Porco},
  {\em Optimal steady-state control for linear time-invariant systems}, in
  {IEEE} Conf.\ on Decision and Control, Dec. 2018, pp.~3251--3257,
  \url{https://doi.org/10.1109/CDC.2018.8619812}.

\bibitem{LSPL-JWSP-EM:21}
{\sc L.~S.~P. Lawrence, J.~W. Simpson-Porco, and E.~Mallada}, {\em
  Linear-convex optimal steady-state control}, IEEE Transactions on Automatic
  Control, 66 (2021), pp.~5377--5384,
  \url{https://doi.org/10.1109/tac.2020.3044275}.

\bibitem{levinson1947perturbations}
{\sc N.~Levinson}, {\em Perturbations of discontinuous solutions of non-linear
  systems of differential equations}, Proceedings of the National Academy of
  Sciences, 33 (1947), pp.~214--218,
  \url{https://doi.org/10.1073/pnas.33.7.214}.

\bibitem{WL-JJES:98}
{\sc W.~Lohmiller and J.-J.~E. Slotine}, {\em On contraction analysis for
  non-linear systems}, Automatica, 34 (1998), pp.~683--696,
  \url{https://doi.org/10.1016/S0005-1098(98)00019-3}.

\bibitem{SM-AH-SB-GH-FD:18}
{\sc S.~Menta, A.~Hauswirth, S.~Bolognani, G.~Hug, and F.~D{\"o}rfler}, {\em
  Stability of dynamic feedback optimization with applications to power
  systems}, in Annual Allerton Conf. on Communication, Control, and Computing,
  2018, pp.~136--143, \url{https://doi.org/10.1109/ALLERTON.2018.8635640}.

\bibitem{pahlevaninezhad2013self}
{\sc M.~Pahlevaninezhad, S.~Eren, P.~K. Jain, and A.~Bakhshai}, {\em
  Self-sustained oscillating control technique for current-driven full-bridge
  {DC/DC} converter}, IEEE Transactions on Power Electronics, 28 (2013),
  pp.~5293--5310, \url{https://doi.org/10.1109/TPEL.2013.2243759}.

\bibitem{LJR-SAB-SSS:16}
{\sc L.~J. Ratliff, S.~A. Burden, and S.~S. Sastry}, {\em On the
  characterization of local {Nash} equilibria in continuous games}, IEEE
  Transactions on Automatic Control, 61 (2016), pp.~2301--2307,
  \url{https://doi.org/10.1109/TAC.2016.2583518}.

\bibitem{MR-RW-IRM:20}
{\sc M.~Revay, R.~Wang, and I.~R. Manchester}, {\em Lipschitz bounded
  equilibrium networks}, arXiv preprint arXiv:2010.01732,  (2020),
  \url{https://doi.org/10.48550/arXiv.2010.01732}.

\bibitem{GR-MDB-EDS:13}
{\sc G.~{Russo}, M.~{Di~Bernardo}, and E.~D. {Sontag}}, {\em A contraction
  approach to the hierarchical analysis and design of networked systems}, IEEE
  Transactions on Automatic Control, 58 (2013), pp.~1328--1331,
  \url{https://doi.org/10.1109/TAC.2012.2223355}.

\bibitem{sannuti1969near}
{\sc P.~Sannuti and P.~Kokotovi{\'c}}, {\em Near-optimum design of linear
  systems by a singular perturbation method}, IEEE Transactions on Automatic
  Control, 14 (1969), pp.~15--22,
  \url{https://doi.org/10.1109/TAC.1969.1099113}.

\bibitem{EDS:07}
{\sc E.~D. Sontag}, {\em Monotone and near-monotone biochemical networks},
  Systems and Synthetic Biology, 1 (2007), pp.~59--87,
  \url{https://doi.org/10.1007/s11693-007-9005-9}.

\bibitem{ART-LM-DN:03}
{\sc A.~R. Teel, L.~Moreau, and D.~Ne{\v{s}}i{\'c}}, {\em A unified framework
  for input-to-state stability in systems with two time scales}, IEEE
  Transactions on Automatic Control, 48 (2003), pp.~1526--1544,
  \url{https://doi.org/10.1109/TAC.2003.816966}.

\bibitem{AT-SF-MP-MHdB-FD:22}
{\sc A.~Terpin, S.~Fricker, M.~Perez, M.~{Hudoba~de~Badyn}, and
  F.~D{\"o}rfler}, {\em Distributed feedback optimisation for robotic
  coordination}, in {A}merican {C}ontrol {C}onference, June 2022,
  \url{https://doi.org/10.23919/acc53348.2022.9867738}.

\bibitem{AT:1948}
{\sc A.~Tikhonov}, {\em On the dependence of the solutions of differential
  equations on a small parameter}, Matematicheskii Sbornik, 64 (1948),
  pp.~193--204.

\bibitem{HT-SJC-JJES:21}
{\sc H.~Tsukamoto, S.-J. Chung, and J.-J.~E. Slotine}, {\em Contraction theory
  for nonlinear stability analysis and learning-based control: {A} tutorial
  overview}, Annual Reviews in Control, 52 (2021), pp.~135--169,
  \url{https://doi.org/10.1016/j.arcontrol.2021.10.001}.

\bibitem{vasil1963asymptotic}
{\sc A.~B. Vasil'eva}, {\em Asymptotic behaviour of solutions to certain
  problems involving non-linear differential equations containing a small
  parameter multiplying the highest derivatives}, Russian Mathematical Surveys,
  18 (1963), p.~13, \url{https://doi.org/10.1070/RM1963v018n03ABEH001137}.

\bibitem{WW-ART-DN:12}
{\sc W.~Wang, A.~R. Teel, and D.~Ne{\v{s}}i{\'c}}, {\em Analysis for a class of
  singularly perturbed hybrid systems via averaging}, Automatica, 48 (2012),
  pp.~1057--1068, \url{https://doi.org/10.1016/j.automatica.2012.03.013}.

\bibitem{yang2021modelling}
{\sc Y.~Yang, S.~Qiao, O.~G. Sani, J.~I. Sedillo, B.~Ferrentino, B.~Pesaran,
  and M.~M. Shanechi}, {\em Modelling and prediction of the dynamic responses
  of large-scale brain networks during direct electrical stimulation}, Nature
  Biomedical Engineering, 5 (2021), pp.~324--345,
  \url{https://doi.org/10.1038/s41551-020-00666-w}.

\end{thebibliography}

\end{document}